\documentclass[12pt]{article}
\usepackage{amsmath,amsthm,amssymb,bm}
\usepackage{booktabs}
\usepackage[margin=2cm]{geometry}
\usepackage{longtable}
\usepackage{pdflscape}
\usepackage{tikz}
\usepackage{xcolor}
\usepackage{tablefootnote}
\usepackage{complexity}
\usepackage{thmtools}
\usepackage{thm-restate}
\usepackage{microtype}
\usepackage[hidelinks]{hyperref}

\newcommand{\Z}{\mathbb{Z}}
\newcommand{\Q}{\mathbb{Q}}
\newcommand{\RR}{\mathbb{R}}
\newcommand{\CC}{\mathbb{C}}
\newcommand{\F}{\mathbb{F}}

\newcommand{\KK}{\mathcal K}
\newcommand{\val}{\nu_\varepsilon}
\newcommand{\uu}{\mathbf{u}}
\newcommand{\vv}{\mathbf{v}}
\newcommand{\UU}{\mathbf{U}}
\newcommand{\VV}{\mathbf{V}}
\newcommand{\APS}{\mathrm{APS}}
\newcommand{\MRV}{\mathrm{MRV}}
\newcommand{\HN}{\mathrm{HN}}
\DeclareMathOperator{\Res}{Res}
\DeclareMathOperator{\TT}{TT}
\DeclareMathOperator{\coeff}{coeff}
\DeclareMathOperator{\im}{im}

\newtheorem{theorem}{Theorem}
\newtheorem{proposition}[theorem]{Proposition}
\newtheorem{corollary}[theorem]{Corollary}
\newtheorem{lemma}[theorem]{Lemma}
\theoremstyle{definition}

\newtheorem{claim}[theorem]{Claim}
\newtheorem*{runningexample}{Running Example}

\title{Approximate Polynomial Satisfiability is in \\the Counting Hierarchy}
\author{Cool Authors}

\newif\ifreviewmarkup
\reviewmarkuptrue

\date{}

\author{
\begin{tabular}{cc}
  \begin{tabular}{c}
    Nikhil Balaji\\
   IIT Delhi, India \\
  \end{tabular}
  &
  \begin{tabular}{c}
    Mahsa Shirmohammadi\\
    CNRS, IRIF, France
  \end{tabular}
  \\[1.5em]
  \begin{tabular}{c}
    S\'ebastien Tavenas\\
    CNRS, LAMA, France
  \end{tabular}
  &
  \begin{tabular}{c}
    James Worrell\\
    University of Oxford, UK
  \end{tabular}
\end{tabular}
}

\begin{document}
\maketitle

\begin{abstract}
The Approximate polynomial satisfiability problem (APS),  introduced by Guo, Saxena, and Sinhababu (CCC 2018), asks whether the zero vector lies in the Zariski closure of the image of a given  polynomial map. Specifically, for a field~$k$ with algebraic closure~$K$, the problem asks  whether
$\bm{0}\in\overline{\bm f(K^n)}$ for a polynomial map
$\bm f=(f_1,\ldots,f_m)$ with $f_i\in k[X_1,\ldots,X_n]$.

APS is a natural topological analogue of Hilbert's Nullstellensatz, namely the question of
whether a given system of polynomial equations has a common zero. APS captures several problems in algebraic complexity, including border rank, hitting sets for border classes, and null-cone membership; it is known to be \NP-hard and in $\PSPACE$.

We show that APS lies in the Counting Hierarchy ($\CH$) over both the rationals and finite fields, substantially improving the known PSPACE upper bound. Our proof builds on a recent breakthrough due to Andrews, Garg, and Schost (FOCS 2026) on deciding Hilbert's Nullstellensatz in $\CH$. As a corollary, our result improves the complexity of certifying hitting sets for border classes from $\PSPACE$ to $\CH$. Such hitting sets play a role in explicit constructions and Noether-normalization problems arising in Geometric Complexity Theory, where a longstanding difficulty has been that algorithmic constructions for $\overline{\VP}$ can be substantially harder than the analogous problem for $\VP$.

We also give a polynomial-time reduction of Hilbert's Nullstellensatz to APS, valid in any characteristic.  In characteristic zero, we give a reduction of APS
to the decision problem for the existential theory of real closed fields.  Overall,
our results place approximate polynomial satisfiability  closer in complexity to exact polynomial feasibility and as a byproduct give improved complexity bounds for several problems arising in approximative complexity. 

\vspace{1em}
\textbf{AI Use: } 
Discussions with  GPT-5.6 Sol and Astra suggested the reduction from HN to APS in Theorem~\ref{thm:redu_HN_APS}  and led to the univariate substitution in Lemma~\ref{lem-weighted-valuation}, which simplified our original reduction from APS to MRV.
All mathematical arguments arising from these interactions were independently checked, developed, and rewritten by the authors, who take full responsibility for the contents and correctness of the paper.
\end{abstract}

\newpage 

\section{Introduction}
\label{sec-intro}

Solving systems of polynomial equations is one of the basic algorithmic problems in algebra and has connections to a remarkably broad range of computational questions. The problem of determining whether 
a given collection of polynomials with coefficients in $k$ has a common zero in $K=\overline{k}$, the algebraic closure of $k$, is called Hilbert's Nullstellensatz, denoted
$\mathrm{HN}_k$. Besides its intrinsic importance in computational algebraic geometry, $\mathrm{HN}_k$ serves as a tool to address many \emph{algebraic existence} problems. A striking example 
is tensor rank: deciding tensor rank over a field $k$ has the same computational complexity as $\mathrm{HN}_k$; in fact, tensor-rank computation and polynomial-system solvability over arbitrary integral domains are known to be polynomial-time equivalent ~\cite{SS-Tensor,shitov}. Similarly, fixing the graph of an arithmetic circuit and varying its scalar labels gives a parameterisation of the polynomials computed by that circuit~\cite{raz-universal}. This readily transforms questions about whether there \emph{exists} a computation of a prescribed form into solving a system of polynomial equations in the circuit parameters which is useful in proving the existence of hard polynomial families~\cite{CKW, SY} as well as designing blackbox identity testing algorithms~\cite{Agrawal, HS80}.

\paragraph{Approximative Complexity}

In this work, we are concerned with approximative or border complexity. In many classical questions in algebraic complexity, one must pass from objects admitting an exact representation of bounded complexity to limits of such objects.  One of the earliest and most influential examples illustrating the power of approximative complexity comes from matrix multiplication.  Bini, Capovani, Lotti, and Romani~\cite{BCLR79,BLR80} showed that allowing arbitrarily small errors can reduce the complexity of bilinear algorithms, and used this to obtain faster matrix
multiplication algorithms. This led to the systematic study of
\emph{border rank}: a tensor has border rank at most $r$ if it belongs to the
closure of the set of tensors of rank at most~$r$.  Lickteig's work in the
1980s developed the algebraic geometry of these closures and used it to prove
lower bounds on border rank~\cite{Lickteig84,Lickteig85}. Studying closure of polynomial maps has been a central theme in algebraic complexity: they occur in approximative circuit complexity, in the factor theorem of B\"{u}rgisser~\cite{Burgisser04}, and especially in Geometric Complexity Theory, where Mulmuley and Sohoni~\cite{MS01,MS08} developed a framework to prove algebraic circuit lower
bounds via non-containment of orbit closures.
For instance, the determinant-versus-permanent~\cite{agrawal-icm} problem can be phrased as the
problem of separating the orbit closure of the determinant polynomial from the orbit
closure of an appropriately padded permanent polynomial.  This reformulation makes
the problem amenable to tools from algebraic geometry, invariant theory, and representation
theory, and has generated a substantial body of work on the geometry of
algebraic computation~\cite{landsberg,Blaser-Ikenmeyer}.

This distinction between exact polynomial maps and their closures leads directly to the problem studied in this paper. 
Let $k$ be either $\Q$ or a finite field $\F_q$, and let $K=\overline{k}$ be an algebraic closure of $k$.
The approximate polynomial satisfiability problem $\APS_k$ asks whether
$\bm{0}\in\overline{\bm f(K^n)}$ for a polynomial map
$\bm f=(f_1,\ldots,f_m)$ with $f_i\in k[X_1,\ldots,X_n]$.   The latter condition is a weakening of the requirement
$\bm{0}\in{\bm f(K^n)}$, i.e., that $f_1,\ldots,f_m$ have a common zero. It is known that $\APS_k$ lies in \PSPACE~\cite{GSS19}, while a celebrated result of Koiran~\cite{koiran} shows that 
 $\mathrm{HN}_{\mathbb Q}$ lies in the polynomial hierarchy, under the generalised Riemann hypothesis.  Recently it was shown that $\mathrm{HN}_k$ lies in the counting hierarchy unconditionally when $k$ is either the  rationals or a  finite field~\cite{AGS26}.

%  In particular, we can think of $\mathrm{APS}_{\mathbb Q}$ as intermediate between
% $\mathrm{HN}_{\mathbb Q}$ and $\mathrm{ETR}$. A key technical contribution of the latter paper is showing the existing of constant-depth exponential-size circuits for the multivariable resultant.  As we observe in Section~\ref{sec-etr}, there are 
%  polynomial-time reductions from $\mathrm{HN}_k$  to $\mathrm{APS}_k$ and 
% from $\mathrm{APS}_{\mathbb Q}$ to the decision problem $\mathrm{ETR}$ for the existential theory of real closed fields.

% The APS problem naturally arises in several different settings in algebraic complexity theory. 
% Given a tensor $T$ over the field $k$ and nonnegative integer $r$,
% whether the border rank of $T$ is at most $r$ can be directly formulated as instance of $\APS_k$, where the sysyem of equations expresses that 
% $T$ can be recovered as the sum of $r$ tensors of rank at most one.   Another application of APS arises in the construction
% of hitting sets for closure of Valiant's class VNP.  Specifically,
% using a universal circuit, one can use APS to
% test whether a proposed hitting set
% for a bounded-degree class of polynomials
% misses a
% nonzero limiting polynomial.  Combining this test with exhaustive enumeration of
% polynomial-size candidate sets gives a deterministic construction of hitting sets
% for bounded-degree polynomials of bounded approximative circuit size.

\medskip

Our main result shows that  $\APS_k$ lies in the counting hierarchy.  The proof naturally divides into two steps.
In the first step we give a randomised polynomial-time reduction of $\APS_k$ to 
a decision problem (called the multivariate resultant valuation $\mathrm{MRV}_k$ problem) 
concerning the resultant of a collection of polynomials with coefficients in the valued field
$\mathcal K:=\overline{K(\varepsilon)}$, that is, the algebraic closure of the field of rational functions in a parameter $\varepsilon$
with coefficients in $K$.  We specify the latter problem precisely in Section~\ref{sec-overview}.  The key fact about the latter problem is that it can 
be decided in the counting hierarchy, using the aforementioned result of~\cite{AGS26}.  This shows that $\APS_k$ also lies in
the counting hierarchy.

\paragraph{Algebraic Independence}    
A closely related problem is \emph{algebraic independence testing}.  
A tuple of
polynomials ${\bm f} = (f_1,\ldots,f_m)\in k[X_1,\ldots,X_n]^m$ is said to be algebraically dependent if there exists a nonzero polynomial $A(Y_1,\ldots,Y_m)\in k[Y_1,\ldots,Y_m]$ 
such that $A(f_1,\ldots,f_m)=0$. We call $A$ an \emph{annihilating polynomial}. Equivalently, writing
\[
I_{\bm f}:=
\{A\in k[Y_1,\ldots,Y_m]:A(\bm f)=0\}
\]
for the ideal of annihilating polynomials, algebraic independence asks whether \(I_{\bm f}=(0)\).  The maximum size of an algebraically independent subtuple is the \emph{transcendence degree} of \(\bm f\). 

The computational complexity of checking whether a given set of polynomials are algebraically dependent crucially depends on the characteristic of the field where the coefficients come from.  Over fields of characteristic zero, and more generally in sufficiently large characteristic, the Jacobian criterion reduces algebraic independence testing to polynomial identity testing, giving a randomised polynomial-time algorithm~\cite{BeeckenMittmannSaxena13}.  In small positive characteristic the Jacobian criterion fails because of \emph{inseparability}.  Guo, Saxena, and Sinhababu showed that over finite fields algebraic independence testing lies in $\AM\cap\coAM$, substantially improving the previously known upper bound~\cite{GSS19}.

Kayal~\cite{Kayal09} investigated the complexity of computing the annhilating polynomial and showed that is $\NP$-hard to test if the constant term of the annhilating polynomial is zero and $\#\P$-hard to evaluate the constant term of the annhilator modulo a given prime $p$. Both these results strongly suggest that the annhilating polynomial is unlikely to be computed by small algebraic circuits. Guo, Saxena, and Sinhababu further showed that $\bm 0\in\overline{\bm f(K^n)}$ if and only if every annihilator of~\(\bm f\) vanishes at the origin: $A(0)=0$ for every $A\in I_{\bm f}$.
Thus algebraic independence testing asks whether there are any algebraic relations among \(f_1,\ldots,f_m\), whereas APS asks whether all such relations vanish at \(\bm 0\). We strengthen Kayal's result to show that this problem is $\NP_{k}$-hard by showing a reduction from $\HN_k$.

\subsection{Applications of APS}
Besides being a natural analogue of Hilbert's Nullstellensatz, APS is versatile in understanding approximative complexity. Here we survey a few such applications.

\paragraph{Hitting sets and Noether normalisation.}
Let $\mathcal C$ be a family of $n$-variate polynomials.  A finite set
$H\subseteq k^n$ is a \emph{hitting set} for $\mathcal C$ if every nonzero
$f\in\mathcal C$ is nonzero at some point of $H$. Constructing explicit hitting sets is a long-standing open problem in complexity theory and is intimately connected with algebraic circuit lower bounds. A classical result of Heintz and Schnorr~\cite{HS80} shows that a random set
of points from a sufficiently large finite grid in $k^n$ is a hitting set with high probability. However, certifying whether a set of points from $k^n$ is a hitting set is a non-trivial task: since arbitrary constants from $k$ are allowed to appear in such circuits in Valiant's model, there
are potentially uncountably many polynomials computable by \emph{small} algebraic circuits, so it is not possible simply to enumerate all polynomials computed by such circuits. To certify a candidate hitting set $H$ for $\VP$, observe that $H$ fails to be a hitting set precisely when there is
a specialisation of a small circuit that computes a non-zero polynomial
vanishing on all points of $H$.  This condition can be encoded by a system of
polynomial equations, and hence by an instance of Hilbert's Nullstellensatz.
This observation gives a $\PSPACE$ procedure to certify a hitting set and using Koiran's theorem,
a $\PH$ construction assuming GRH.  In view of the recent result
$\HN_k\in\CH$~\cite{AGS26}, the same certification can now be
performed unconditionally in the Counting Hierarchy.

A natural question that arises in this context is the complexity of certifying hitting sets for the border class $\overline{\VP}$. A polynomial
in $\overline{\VP}$ is not known to admit small algebraic circuits: it is, however, a limit
of polynomials having small circuits. A striking example of the power of approximative complexity is the case of width-2 algebraic branching programs. It is known that this model is not universal, i.e. there exist explicit polynomials that cannot be computed by width-2 ABPs~\cite{Allender-Wang}. However, the corresponding border class is known to be universal~\cite{BIZ}. Thus, a set that is a hitting set for every member
of an approximating family need not be a hitting set for its limit, and the preceding discussion about certifying hitting sets for $\VP$ does not directly translate to the analogous question for $\overline{\VP}$. The work of Forbes and Shpilka~\cite{ForbesShpilka18} and Guo, Saxena, and Sinhababu~\cite{GSS19} addressed precisely this issue. Forbes and Shpilka~\cite{ForbesShpilka18} gave a
PSPACE upper bound to certify hitting sets for $\overline{\VP}$ via the notion of \emph{robust hitting sets}. Guo, Saxena, and Sinhababu~\cite{GSS19} subsequently generalised the $\PSPACE$ bound to hold over arbitrary fields, via a reduction to APS. Both works~\cite{ForbesShpilka18, GSS19} raise the question of whether the complexity of certifying hitting sets for $\overline{\VP}$ can be improved to $\PH$.

Hitting sets also arise in Mulmuley's formulation of
Noether normalisation in Geometric Complexity Theory~\cite{Mulmuley17}. Mulmuley shows~\cite{Mulmuley17}
that for any explicit families of varieties (and specifically for the determinantal variety $\Delta[\det,m]$), the construction of a \emph{separating
explicit system of parameters} reduces to the construction of hitting sets.
Mulmuley observed a substantial gap between the complexity of the underlying
algebraic objects and the then-known complexity of constructing such
normalizing maps, calling this the \emph{GCT chasm}.  The PSPACE upper bounds for certifying hitting sets due to
Forbes–Shpilka and Guo–Saxena–Sinhababu substantially narrowed this gap.
Our result further mitigates the GCT chasm to $\CH$.

\paragraph{Null Cone Problem}

A basic closure problem in computational invariant theory is \emph{null-cone membership}. Let a reductive algebraic group \(G\) act linearly on a vector space \(V\).  The null cone is
\[
    \mathcal N_G(V)
      :=\{v\in V:0\in\overline{G\cdot v}\}.
\]
This definition has several equivalent interpretations which make the problem particularly rich.  In classical invariant theory, \(\mathcal N_G(V)\) is precisely the common zero set of all homogeneous \(G\)-invariant polynomials of positive degree.  Thus null-cone membership is a highly structured analogue of polynomial-system feasibility, where the defining equations are given only implicitly through the group action rather than explicitly as an input system. 

The Hilbert–Mumford criterion says that \(v\in\mathcal N_G(V)\) if and only if there is a one-parameter subgroup $\lambda:\mathbb G_m\to G$
such that $\lim_{t\to0}\lambda(t)\cdot v=0$. Thus membership in the null cone is already a statement about algebraic degeneration: although no group element sends a nonzero \(v\) exactly to zero, a suitable one-parameter family may do so in the limit.

Over \(\CC\), there is a complementary analytic formulation.  If $\operatorname{cap}(v):=\inf_{g\in G}\|g\cdot v\|$, then $v\in\mathcal N_G(V)$ if and only if 
$\operatorname{cap}(v)=0$. This optimisation-theoretic point of view has been algorithmically fruitful. B\"urgisser, Garg, Oliveira, Walter, and Wigderson~\cite{BGOW18} showed that null-cone membership captures noncommutative rank and identity testing over the free skew field with several connections to problems in noncommutative algebra, combinatorial optimisation, quantum information, and Geometric Complexity Theory.

 For a fixed vector \(v\), the orbit map $G\to V$, given by $g \mapsto g\cdot v$ is an algebraic morphism, and the condition~$0\in\overline{G\cdot v}$
can therefore be framed as an image-closure membership problem.  Guo, Saxena, and Sinhababu observed that the null-cone problem considered in~\cite{BGOW18} can be formulated as an instance of APS and consequently obtained a $\PSPACE$ upper bound~\cite{GSS19}, improving the previously known exponential-space bound. Independently, \cite{BGOW18} give an exponential-time algorithm for this problem.

\section{Main Results and Proof Overview}
\label{sec-overview}
\paragraph{Background.} Recall that an algebraic closure of a field $k$ is an algebraic extension
$\overline{k}$ of $k$ that is algebraically closed. This means that every
element of $\overline{k}$ satisfies a nonzero polynomial equation with
coefficients in $k$, and every nonconstant polynomial with coefficients in
$\overline{k}$ has a root in $\overline{k}$.  In this paper, $k$ is either the
field $\Q$ of rational numbers or a finite field $\F_q$, and
$K:=\overline{k}$ denotes an algebraic closure.

Denote by $K(\varepsilon)$ the field of rational functions in the formal parameter $\varepsilon$ with coefficients in~$K$.
This field carries a valuation $\val:K(\varepsilon)\to\mathbb Z\cup\{\infty\}$ which intuitively gives the order
to which the parameter~$\varepsilon$ divides $f$.  Formally, $\val(f):=j_0$ for a polynomial $f=\sum_{j=j_0}^d a_j\varepsilon^j$,
where $a_{j_0}\neq 0$. Moreover, $\val(f/g):=\val(f)-\val(g)$ for non-zero polynomials $f,g\in K[\varepsilon]$; and 
$\val(0):=\infty$.  The valuation $\val$ can be extended to a valuation on the algebraic closure 
$\KK:=\overline{K(\varepsilon)}$ that takes values in~$\Q$.  We choose such an extension and keep the notation $\val$.
We say that $b$ has positive valuation if
$\val(b)>0$.
Recall also the basic properties
\[
\val(ab)=\val(a)+\val(b)\quad\text{and}\quad 
\val(a+b)\geq\min\{\val(a),\val(b)\},
\]
where 
equality holds in the addition relation whenever the terms of least valuation
do not cancel. 

Below, we  follow the convention that polynomial variables are written in upper
case, whereas  values are written in lower case. Unless a different
range is specified, $\bm X=(X_1,\ldots,X_n)$ and
$\bm Y=(Y_1,\ldots,Y_{n+1})$ denote tuples of variables;
their values $\bm x$ and $\bm y$ are tuples of elements of the relevant
ring. For
$\bm\alpha\in\Z_{\geq0}^n$, denote by
$\bm X^{\bm\alpha}$ the monomial~$X_1^{\alpha_1}\cdots X_n^{\alpha_n}$.

We  write $[m]=\{1,\ldots,m\}$ and $[s,m]=\{s,s+1,\ldots,m\}$, with the
latter set empty when~$s>m$. 
\paragraph{The APS Problem.}
For a field $k$ as above,
the approximate polynomial satisfiability problem ($\APS_k$ for short) is a relaxation
of polynomial satisfiability, also called the Hilbert Nullstellensatz
problem ($\mathrm{HN}_k$). Let
\[
\bm f=(f_1,\ldots,f_m):K^n\longrightarrow K^m,
\qquad (\text{with } f_i\in k[\bm X]),
\]
be a polynomial map.  We use standard conventions for representing polynomials; see Section~\ref{sec-prelim}
for details.
The $\mathrm{HN}_k$ problem asks whether there exists a common zero
$\bm x^*\in K^n$ of the polynomials~$f_1,\ldots,f_m$; equivalently, it asks
whether $\bm0\in\bm f(K^n)$. The $\APS_k$ problem replaces the image of~$\bm f$ by its Zariski closure and asks instead whether
\[
\bm0\in\overline{\bm f(K^n)}.
\]
Clearly, every positive instance of $\mathrm{HN}_k$ is a positive instance of $\APS_k$.
On the other hand, a system may be a positive instance of $\APS_k$ without having a common zero:
for example, $\bm f=(XY-1,X)$ has no common zero, but
\[
\im(\bm f)=\{(z_1,z_2)\in K^2:z_2\neq0\}\cup\{(-1,0)\}
\]
is dense in $K^2$, so $\bm0\in\overline{\im(\bm f)}$.

Both $\mathrm{HN}_k$ and $\APS_k$ admit purely algebraic formulations, respectively in
terms of the ideal and the sub-$k$-algebra generated by the $f_i$'s within $k[\bm X]$.
The weak Nullstellensatz says that the system
$\{f_i=0:i\in[m]\}$ has no common zero in $K^n$ exactly when the ideal in
$k[\bm X]$ generated by the $f_i$'s contains $1$. Explicitly, $\bm f$ is a
negative instance of $\mathrm{HN}_k$ if and only if there exist $g_i\in k[\bm X]$ such
that $\sum_{i=1}^m g_if_i=1$.
On the other hand, $\bm f$ is a negative instance of $\APS_k$ exactly when
there is an annihilating polynomial of $f_1,\ldots,f_m$ with nonzero
constant term. Explicitly, this says that there exist nonzero vectors
$\bm\alpha_1,\ldots,\bm\alpha_s\in\Z_{\geq0}^m$ and coefficients
$c_1,\ldots,c_s\in k$ such that
\[
\sum_{i=1}^s c_i f_1^{\alpha_{i,1}}\cdots f_m^{\alpha_{i,m}}=1.
\]

We will make extensive use of an equivalent formulation of
$\APS_k$ in terms of the valued field $\KK$: namely, an instance is positive 
precisely when there exists $\bm x^*(\varepsilon)\in\KK^n$ such that
\[
\val\bigl(f_i(\bm x^*(\varepsilon))\bigr)>0\qquad(\text{with }i\in[m]).
\]
The equivalence of the above with the image-closure formulation is valid in arbitrary
characteristic. Indeed, the curve formulation of APS in
\cite[Theorem~4.2]{GSS19} gives a witness $\bm x^*(\varepsilon)$ with coordinates in
$K[\varepsilon,\varepsilon^{-1}]$ whenever the origin belongs to the image closure. Conversely,
if a valued-field witness exists, substituting it into any annihilator and
reducing in the valuation ring shows that the annihilator has zero constant
term. We can think of the image of $\bm x^*(\varepsilon)$ as a parametrised algebraic
curve having $\bm0$ in its Zariski closure. Thus an HN witness makes all input
polynomials exactly zero, whereas an APS witness makes all their values
infinitesimal.

\paragraph{Statement of Main Results.}
Our main result states that $\APS_k$ can be decided in the counting hierarchy:
\begin{restatable}{theorem}{mainTheorem}
\label{theo-main}
Let $k$ be either $\Q$ or a finite field $\F_q$.  Then $\APS_k$ lies 
in the counting hierarchy. 
\end{restatable}

To explain the two main steps of the proof, we introduce a second decision problem:
\emph{the multivariate resultant valuation problem} $\mathrm{MRV}_k$.
The input to this problem consists of a system of polynomials 
$\bm P=(P_0,\ldots,P_n)\in k[\varepsilon,\bm Y][X_0,\ldots,X_n]^{n+1}$,
where $\bm Y=(Y_1,\ldots,Y_v)$ and each $P_i$ is homogeneous of
positive degree in $X_0,\ldots,X_n$.  The input also specifies
a monomial $\bm Y^{\bm\beta_0}$ with~$\bm\beta_0\in\Z_{\geq0}^v$.
Writing
\[
\Res(\bm P)=\sum_{\bm\beta}c_{\bm\beta}(\varepsilon)\bm Y^{\bm\beta},
\]
the question is whether
\[
\val(c_{\bm\beta_0}(\varepsilon))>
\min_{\bm\beta}\val(c_{\bm\beta}(\varepsilon)).
\]
%Clearly the minimum on the right-hand side above just depends on the non-zero coefficients
%of the resultant.

The proof of Theorem~\ref{theo-main} relates the APS and MRV problems as follows:
\begin{restatable}{theorem}{maintheoreduc}
\label{theo-main-reduc}
There is a randomised polynomial-time Turing
reduction from $\APS_{\mathbb Q}$ to $\mathrm{MRV}_{\mathbb Q}$.
For a finite field $k$, there is a randomised polynomial-time Turing
reduction from $\APS_k$ to $\mathrm{MRV}_\ell$, where~$\ell$ is a finite extension of $k$
of degree polynomial in the input length.
\end{restatable}
To complete the proof of Theorem~\ref{theo-main}
we show that $\mathrm{MRV}_k$ lies in the counting hierarchy.  This step uses recent results of~\cite{AGS26} on the existence of uniform
constant-depth arithmetic circuits for the resultant.

In the remainder of this section we give an overview of the proofs of 
Theorems~\ref{theo-main} and~\ref{theo-main-reduc}.  Full details are given
in Sections~\ref{sec-algebaric} and~\ref{sec-algorithm}.

\paragraph{Algebraic Independence of Input Polynomials.}
Recall that polynomials $f_1,\ldots,f_n\in k[\bm X]$ are algebraically independent if 
there is no nonzero
polynomial $P\in k[Y_1,\ldots,Y_n]$ such that $P(f_1,\ldots,f_n)=0$, that is,
the
annihilator ideal of the $f_i$'s is the zero ideal. Since the annihilator
ideal defines the Zariski closure of the image of
$(f_1,\ldots,f_n):K^n\to K^n$, algebraic independence is equivalent to this
image being Zariski dense in~$K^n$.
The transcendence degree of a family of polynomials is the maximum cardinality 
of an algebraically independent subset.

For example, the coordinate polynomials $X_1,X_1X_2,X_1^2$ have
transcendence degree $2$: $X_1$ and $X_1X_2$ are algebraically independent,
whereas $X_1$ and $X_1^2$ satisfy the relation $Y_1^2-Y_2=0$.
Generally speaking, any collection of polynomials in $n$ variables has
transcendence degree at most $n$, since they all belong to $k(\bm X)$,
whose transcendence degree over $k$ is $n$. However, having transcendence degree
$n$ does not imply that every proper subset of $n+1$ coordinate
polynomials is algebraically independent.

Let $\bm g=(g_1,\ldots,g_m):K^s\to K^m$ be an instance of $\APS_k$.
The random-projection result of~\cite[Theorem 4.6]{GSS19} can be used to (with high probability)
transform such an instance
into one in which the number of
coordinate polynomials is one more than their transcendence degree, without changing the APS status. 
We call a map $\bm f=(f_1,\ldots,f_{n+1}):K^n\to K^{n+1}$
\emph{minimally dependent} if every proper subset of its coordinate
polynomials is algebraically independent. We need a strengthening of the
reduction in~\cite{GSS19} in which the resulting map is minimally
dependent.

\begin{restatable}{proposition}{propsubsetsnindependent}
\label{prop-subsets-n-independent}
There is a randomised polynomial-time reduction from the general
$\APS_k$ problem to the version for minimally dependent instances. Over a finite field,
the output coefficients may lie in a finite extension $\ell/k$ of
polynomial degree; over $\Q$, they remain rational.
\end{restatable}

The reduction first adjoins the polynomials $Y,YX_1,\ldots,YX_s$ to the input.
This preserves the APS status: a negative instance  remains negative after these extra polynomials are added, while
a positive instance stays positive by choosing $Y$ to be a sufficiently high power of $\varepsilon$.
The crucial point is that this modification makes the transcendence degree of the map equal to the number
$s+1$ of  variables.  A generic projection onto $s+2$ coordinates
then yields a minimally dependent map with high probability: the Jacobian 
associated with this projection has a nonzero minor in every characteristic,  
showing that every $(s+1)$-subtuple of polynomials is algebraically independent.
We give a  detailed proof  in Section~\ref{sec-algebaric}.

Throughout the remainder of the argument, we work with minimally
dependent instances of the APS problem, consisting of $n+1$ polynomials in $n$ variables.  
For such an instance $\bm f$, 
the image closure associated with every proper subset
of the coordinate polynomials contains the origin (since such a subset is algebraically independent). The remaining question
is whether the origin also belongs to the image closure of the full map
$\bm f$.

Since the transcendence degree of $f_1,\ldots,f_{n+1}$ is $n$, their
annihilator ideal is principal. Hence there is a unique irreducible
polynomial $H\in k[\bm Y]$, up to scaling by a nonzero constant, that
generates this ideal and satisfies
$H(f_1(\bm X),\ldots,f_{n+1}(\bm X))=0$.
We refer to $H$ as the \emph{irreducible annihilator} of the $f_i$'s. In this setting
the input instance is positive if and
only if $H$ vanishes at the origin. Indeed, the image closure is defined
by $H$, so
\[
\bm0\in\overline{\bm f(K^n)}
\quad\Longleftrightarrow\quad H(\bm0)=0.
\]
Moreover, since the annihilator ideal is generated by $H$, $H(\bm0)=0$ if and only if 
every polynomial relation among the $f_i$'s vanishes at the
origin. 
%The description remains valid after extending scalars to $K$;
%see Claim~\ref{claim:ONE}.

\paragraph{Prescribing the Algebraic Curve.}
We now describe the algebraic curve that we use as a witness for positive
APS instances. Consider the annihilator $H\in k[\bm Y]$, and suppose
$H(\bm0)=0$. Let $d$ be the least total degree of a monomial appearing in
$H$ and denote by $H_d$ the homogeneous degree-$d$
part of $H$. By assumption $H_d$ is not constant, so by renaming the input polynomials if necessary, 
we may assume that $Y_{n+1}$
appears in $H_d$.

By algebraic independence, the image closure of the first $n$ polynomials
$(f_1,\ldots,f_n):K^n\to K^n$ is $K^n$.
Roughly speaking, the idea is to prescribe infinitesimal values for these
polynomials while leaving the value of $f_{n+1}$ unspecified; that is, we set
\begin{equation}
\label{eq-imageprescribed}
f_i(\bm X)=a_i\varepsilon \qquad(\text{with }i\in[n]).
\end{equation}
For a generic choice of
$\bm a$, the system in~\eqref{eq-imageprescribed} has a nonempty
zero-dimensional set of solutions over~$\KK$.  We will show that 
when the $a_i$'s are  chosen at random from a sufficiently large finite
set $S \subseteq K$, 
these finiteness and nonemptiness conditions hold with high probability.

The above discussion is subject to our assumption that $Y_{n+1}$ appears in the least-degree
homogeneous component $H_d$ of the annihilator $H$.  Our algorithm will consider all $n+1$
possibilities for the choice of the polynomial that is left out of the system~\eqref{eq-imageprescribed}, i.e.,
the polynomial whose value is left unspecified.

Our main objective is to
show that, whenever the APS instance is positive, the $a_i$'s in~\eqref{eq-imageprescribed} can be
chosen so that among the finitely many solutions of the system there is some
$\bm x^*(\varepsilon)\in \KK^n$ satisfying
\begin{equation}
\label{eq-vallastf}
\val\bigl(f_{n+1}(\bm x^*(\varepsilon))\bigr)>0.
\end{equation}
Together, \eqref{eq-imageprescribed} and~\eqref{eq-vallastf} give an APS
witness, since all the $f_i$'s have positive valuation at $\bm x^*(\varepsilon)$.

The following proposition shows that the $a_i$'s that fail to satisfy the above desiderata
are roots of a polynomial of a controlled degree.
\begin{restatable}{proposition}{propchoiceaipositiveimage}
\label{prop-choice-ai-positive-image}
Let $\bm f=(f_1,\ldots,f_{n+1})\in k[\bm X]^{n+1}$ be minimally
dependent, and put
\[
\Delta:=\max_{i\in[n+1]}\deg f_i,
\qquad \Delta_{\mathrm{bad}}:=\Delta^{n+1}+n\Delta^n.
\]
There is a nonzero polynomial $Q_{\mathrm{bad}}\in k[Y_1,\ldots,Y_n]$ of degree at most
$\Delta_{\mathrm{bad}}-1$ such that, after a suitable permutation of the
$f_i$'s, for every $\bm a\in K^n$ with $Q_{\mathrm{bad}}(\bm a)\neq0$:
\begin{enumerate}
\item the set $V\subseteq\KK^n$ of common zeros of
$\{f_i-a_i\varepsilon:i\in[n]\}$ is finite and nonempty; and
\item $\bm f$ is a positive APS instance if and only if there exists
$\bm x(\varepsilon)\in V$ such that
$\val(f_{n+1}(\bm x(\varepsilon)))>0$.
\end{enumerate}
\end{restatable}

For a finite set $S\subseteq K$, the probability that a
uniformly chosen $\bm a\in S^n$ violates $Q_{\mathrm{bad}}(\bm a)\neq0$ is at most
$(\Delta_{\mathrm{bad}}-1)/|S|$. 
In case $k=\mathbb Q$ we choose $S$ to be a finite set of integers, whereas when $k=\mathbb F_q$ then 
we choose $S$ to be a subset of a suitable extension field $\ell$ of $k$.
We thereby obtain randomised polynomial-time sampling
with an arbitrarily small constant failure probability. The polynomials defining the bad choices
are only used to obtain this error bound; the algorithm does not compute them.

Before outlining the proof of
Proposition~\ref{prop-choice-ai-positive-image}, we introduce our running
example.
\begin{runningexample}
Consider the minimally dependent map $\bm f=(f_1,f_2,f_3)$, where
\begin{align*}
f_1&:=X_1(X_1+1),\\
f_2&:=X_2(X_1+1),\\
f_3&:=X_1(X_2+1).
\end{align*}
Its annihilator
\[
H:=(Y_1-Y_2)(Y_1-Y_3)-Y_1(Y_2-Y_3)^2
\]
has constant term zero, regardless of the choice of the field $k$.
Now prescribe $f_1=a_1\varepsilon$ and $f_2=a_2\varepsilon$, with $a_1\neq0$. 
Writing $\mathcal R:= \{ x \in \mathcal K:\val(x)\geq 0\}$ for the valuation ring of 
$\mathcal  K$, 
there is a solution
\[
\bm x^*(\varepsilon) \in\bigl(a_1\varepsilon+\varepsilon^2\mathcal R,\ a_2\varepsilon+\varepsilon^2\mathcal R\bigr)
\]
whose image under $f_3$ lies in $a_1\varepsilon+\varepsilon^2\mathcal R$ and has positive valuation.
The second solution is
\[
\bm x^+(\varepsilon) \in \bigl(-1-a_1\varepsilon+\varepsilon^2\mathcal R,\ -a_2/a_1+\varepsilon\mathcal R\bigr).
\]
Its image under $f_3$ lies in $a_2/a_1-1+\varepsilon\mathcal R$ and therefore has valuation zero
for generic $a_1,a_2$. 

Thus, after prescribing the values of $f_1$ and $f_2$, the possible values of
$f_3$ include two series, generically only one of which has positive
valuation. The question is how to show that a value of $f_3$ that has positive valuation
always exists for a generic choice of parameters in a positive instance of APS.
\end{runningexample}

\paragraph{Proof Outline of Proposition~\ref{prop-choice-ai-positive-image}.}
In the case of a positive instance of $\APS_k$,
we prove the existence of $\bm x(\varepsilon)$ satisfying~\eqref{eq-imageprescribed} and~\eqref{eq-vallastf}
using $H\in k[\bm Y]$ and fundamental facts from commutative algebra.
Consider $H$ evaluated at the prescribed values of the first $n$
polynomials:
\begin{equation}
\label{eq-Himages}
H(a_1\varepsilon,\ldots,a_n\varepsilon,Z).
\end{equation}
We want to show that this polynomial has a root $z^*(\varepsilon)$ of positive
valuation. Recall that $d$ is the least total degree of a monomial occurring in $H$.
For a fresh indeterminate $U$, define
\[
H^{\circ}(\varepsilon,U):=\varepsilon^{-d}H(a_1\varepsilon,\ldots,a_n\varepsilon,\varepsilon U)\in K[\varepsilon][U].
\]
Then $H^{\circ}(0,U)=H_d(a_1,\ldots,a_n,U)$, which is nonconstant in $U$ for a
generic choice of $\bm a$, since $Y_{n+1}$ occurs in $H_d$.
A valuation argument shows that $H^{\circ}(\varepsilon,U)$, for generic choice of~$\bm{a}$, has a root~$u^*(\varepsilon)$ of
nonnegative valuation. Consequently $z^*(\varepsilon):=\varepsilon u^*(\varepsilon)$ is a root of
\eqref{eq-Himages} of strictly positive valuation. Our argument uses a
coefficient of a positive power of $Y_{n+1}$ in $H_d$,  and  also applies in positive characteristic;
see Claim~\ref{claim:TWO} for the detailed proof.

The preceding argument gives a positive-valuation root of the specialised
annihilator. We next show that every such root occurs as the image of
$f_{n+1}$ at a solution of~\eqref{eq-imageprescribed}. For this, we use the
following result on extension of homomorphisms along integral ring extensions; see
\cite[Chapter~VII, Proposition~3.1]{lang2012algebra}.  Recall here that
an element $b$ of a ring $B$ is integral over a subring~$A$ if it is a root of a monic
polynomial with coefficients in~$A$, that is, if
$b^d+\sum_{i=0}^{d-1}a_ib^i=0$ for some $a_0,\ldots, a_{d-1}\in A$. The ring $B$ is
integral over $A$ if every element of $B$ is integral over $A$.
Intuitively, the elements of $B$ are algebraically constrained by those
of $A$.  We now have:

\begin{proposition}
\label{prop-lang}
Let $A$ be a subring of a commutative ring $B$ with identity, and assume
that $B$ is integral over $A$. Let $\varphi:A\to L$ be a homomorphism to
an algebraically closed field $L$. Then there exists a
homomorphism $\widetilde \varphi: B\to L$ such that the following diagram commutes:
\begin{align*}
A&\hookrightarrow B\\
\varphi\,\downarrow&\qquad\downarrow\,\widetilde\varphi\\
L&~=~L
\end{align*}
\end{proposition}

We apply this proposition to the rings
\[
%A:=K[f_1,\ldots,f_n],\qquad
B:=K[f_1,\ldots,f_n,f_{n+1}]\quad\text{and}\quad
C:=K[X_1,\ldots,X_n].
\]
For any root $z^*(\varepsilon)$ of~\eqref{eq-Himages}, there is a unique 
$K$-algebra  homomorphism $\varphi:B\to\KK$ such that
\begin{gather}
\varphi(f_i)=a_i\varepsilon \quad(\text{with }i\in[n]),
\qquad \varphi(f_{n+1})=z^*(\varepsilon)\, .
\label{eq:def-phi}
\end{gather}
Indeed, all polynomial
relations among $f_1,\ldots,f_{n+1}$ are multiples of~$H$.  Every such relation is preserved by 
the assignments in~\eqref{eq:def-phi} that define $\varphi$, since
\[ \varphi(H(f_1,\ldots,f_{n+1}))=
H(a_1\varepsilon,\ldots,a_n\varepsilon,z^*(\varepsilon))=0 \, .\]

We cannot directly apply Proposition~\ref{prop-lang} to extend $\varphi$ along the inclusion $B\subseteq C$
since $C$ need not be integral over $B$. We use localisation to overcome
this obstacle, that is, we formally add an inverse of a certain element of $B$.
Since $f_1,\ldots,f_n$ are algebraically independent,
$k(\bm X)$ is a finite algebraic extension of the field $k(f_1,\ldots,f_n)$; therefore each $X_j$  satisfies a
monic equation with coefficients in $k(f_1,\ldots,f_n)$. Let  polynomial~$s\in k[Y_1,\ldots,Y_n]$  be a
nonzero polynomial that clears all the
denominators of such equations.  We have $s(f_1,\ldots,f_n)\neq0$ by algebraic independence of polynomials~$f_1,\ldots,f_n$.

We localise the rings  $B$ and $C$ by inverting $s(f_1,\ldots,f_n)$, and denote
the resulting rings by $B_s$ and $C_s$, respectively. Recall that the elements of the localisation are formal fractions whose denominators 
are powers of $s(f_1,\ldots,f_n)$.
Notice that every $X_j$ now satisfies a monic
equation over~$B_s$, so~$C_s$ is integral over~$B_s$.  Moreover,
if $s(a_1\varepsilon,\ldots,a_n\varepsilon)\neq0$, the homomorphism $\varphi$ extends from $B$ to
$B_s$ in a natural way, and then Proposition~\ref{prop-lang} further extends this map to~$\widetilde\varphi:C_s\to\KK$.

Define $\bm{x}^*(\varepsilon) \in \KK^n$ by $x_j^*(\varepsilon):=\widetilde\varphi(X_j)$ for all $j\in[n]$.
Then $\widetilde\varphi$ acts on a rational function $f \in C_s$ by evaluating it 
at $\bm{x}^*(\varepsilon)$.  
Since 
$\widetilde\varphi$ agrees with $\varphi$ on $B$, we obtain
\[ f_i(\bm x^*(\varepsilon)) = \widetilde\varphi (f_i) = \varphi(f_i)= a_i \varepsilon \quad (\text{with }i\in[n])
\]
and 
\[ f_{n+1}(\bm x^*(\varepsilon)) = \widetilde\varphi (f_{n+1}) = \varphi(f_{n+1})= z^*(\varepsilon)\, .\]
Thus, for generic $\bm a$, the roots of
$H(a_1\varepsilon,\ldots,a_n\varepsilon,Z)$ are exactly the images under $f_{n+1}$ of the
solutions of the prescribed system. 

\paragraph{Multivariate Resultants and $f_{n+1}$-Images.}
Assume that a generic choice of $\bm a$ has been fixed, so the system
$\{f_i-a_i\varepsilon :i\in[n]\}$ has a non-empty finite set $V$ of solutions and, for a positive
APS instance, at least one solution gives $f_{n+1}$ positive valuation.
Since $V$ may have cardinality exponential in the size of the above system, in order to verify the latter condition
we need a succinct representation
of the values of $f_{n+1}$ at all of the points of $V$.
To this end, we construct a polynomial $R$ in a single variable $Z$ whose roots are the values 
of $f_{n+1}$ on $V$. We will use the facts on resultants in~\cite{AGS26} to access
the coefficients of $R$ in $\CH$.  This will allow us to determine whether at least one of the roots 
of $R$ has strictly positive valuation.

Write $G_i(X_0,\ldots,X_n)$ for the homogenisation, in the $\bm X$
variables, of $f_i-a_i\varepsilon$, for all~$i\in[n]$. Write $F$ for the
homogenisation of $f_{n+1}$, and $d$ for its degree. The coefficients of
these forms lie in~$\ell[\varepsilon]$, where~$\ell$ is the field extension of $k$
generated by $a_1,\ldots,a_n$.
If the projective zero set of the~$G_i$'s has no point at infinity,
the Poisson formula~\cite[Theorem 3.4]{cox1992ideals} expresses the resultant
$R$ of the $G_i$'s and $ZX_0^d-F$ as
\begin{equation}
\label{eq-Rz}
R(Z)=c\prod_{\bm x(\varepsilon)\in V}
\bigl(Z-f_{n+1}(\bm x(\varepsilon))\bigr)^{m_{\bm x}},
\qquad c\in\KK^\times,
\end{equation}
where the positive integers $m_{\bm x}$ record the multiplicity of $\bm x$ in $V$.

As noted in~\cite{AGS26}, points at infinity can cause this
resultant to vanish identically. To repair this, we use the generalised
characteristic polynomial construction of~\cite{Ierardi89}, in a form
similar to~\cite{AGS26}. First perturb the homogenised polynomials $G_i$ with
a new indeterminate $T$:
\[
\widetilde G_i:=G_i+TX_i^{d_i},\qquad \text{where } d_i:=\deg_{X_0,\ldots,X_n}G_i.
\]
For a further indeterminate $W$, put
\[
L:=ZX_0^d+\sum_{i=1}^nW^iX_i^d-F
\]
and define
\begin{equation}
\label{eq-R0}
R_0(T,W,\varepsilon,Z):=\Res(\widetilde G_1,\ldots,\widetilde G_n,L)
\end{equation}
where the resultant is taken with respect to $X_0,\ldots,X_n$.

For a non-zero polynomial $P=\sum_{j=0}^d p_jU^j$ in an indeterminate $U$, we write $\TT_U P$
for the trailing coefficient, that is, the coefficient $p_{j_0}$ where $j_0:=\min\{j:p_j\neq 0\}$.
By convention we define~$\TT_U(0)=0$.

We will show that
\begin{gather}
R(Z)=R_2(\varepsilon,Z):=\TT_W\TT_T R_0(T,W,\varepsilon,Z)    
\label{eq:iterated}
\end{gather}
has the form~\eqref{eq-Rz}.

By~\cite[Lemma~2.6]{Ierardi89} and
\cite[Lemmas~4.6--4.7]{AGS26}, the $T$-trailing coefficient is a product
of factors obtained by evaluating $L$ at a finite set of projective
limit points of the perturbed system. Every isolated affine solution of
$\{f_i-a_i\varepsilon :i\in[n]\}$ occurs among these points. At an affine point~$[1:x_1:\cdots:x_n]$, the corresponding factor is a nonzero scalar
multiple of
\[
Z+\sum_{i=1}^nW^ix_i^d-f_{n+1}(\bm x),
\]
whose $W$-trailing coefficient is $Z-f_{n+1}(\bm x)$.
At {a point at infinity $[0:x_1:\cdots:x_n]$, the factor is a nonzero scalar multiple of}
\[
\sum_{i=1}^nW^ix_i^d-F(0,x_1,\ldots,x_n).
\]
It is a nonzero polynomial in $W$, since the $W^i$'s are distinct and
some $x_i\neq0$. Its $W$-trailing coefficient is a nonzero scalar,
independent of $Z$.
%\textcolor{red}{Mahsa: next sentence needs improvement, it is a bit out of blue}
%No separability assumption is used here: the resultant factorisation is taken with multiplicity, so the argument remains valid in arbitrary characteristic, including at inseparable points.

\begin{runningexample}
The homogenisations of $f_1-a_1\varepsilon$ and $f_2-a_2\varepsilon$ are
\begin{align*}
G_1&=X_1^2+X_0X_1-a_1\varepsilon X_0^2,\\
G_2&=X_1X_2+X_0X_2-a_2\varepsilon X_0^2.
\end{align*}
There is a common root  $[0:0:1]$ of $G_1$ and $G_2$ at infinity. The homogenisation
$F=X_1X_2+X_0X_1$ of $f_3$ also vanishes there, so
$\Res(G_1,G_2,ZX_0^2-F)$ is identically zero. Our approach considers the resultant of 
the perturbed polynomials $\widetilde G_1,\widetilde G_2$ and
\[
L=ZX_0^2+WX_1^2+W^2X_2^2-F.
\]
For generic $a_1,a_2$, the point at infinity has multiplicity two, and
\begin{align*}
\TT_T\Res(\widetilde G_1,\widetilde G_2,L)
={}&cW^4
\bigl(Z+W(x_1^*)^2+W^2(x_2^*)^2-f_3(\bm x^*)\bigr)\\
&\hspace{12mm}\cdot
\bigl(Z+W(x_1^+)^2+W^2(x_2^+)^2-f_3(\bm x^+)\bigr)
\end{align*}
for some $c\in\KK^\times$. Taking the $W$-trailing coefficient gives
$c\bigl(Z-f_3(\bm x^*)\bigr)\bigl(Z-f_3(\bm x^+)\bigr)$, as required.
\end{runningexample}

\paragraph{$\varepsilon$-Valuations of $R(Z)$ and $R(0)$.}
It remains to determine whether $R(Z)$ has a root of positive valuation.
%\textcolor{red}{Mahsa: do we name this gauss valuation? if yes, we need to drop a word about its relation with $\val$?}
For a polynomial $P\in\KK[Z]$, define its \emph{Gauss valuation}
to be the minimum $\val$-valuation of its coefficients. 
Overloading notation, we denote the Gauss valuation by $\val(P)$, with
$\val(0):=\infty$. This valuation is multiplicative: after scaling
two nonzero polynomials to have valuation zero, their reductions are
nonzero polynomials over a residue field, whose product is nonzero.

Write
\begin{equation}
\label{eq-Rzval}
R(Z)=\sum_{i=0}^{D_R}c_i(\varepsilon)Z^i,
\end{equation}
where $D_R$ is the number of factors in~\eqref{eq-Rz}, counting multiplicity.
Then $R$ has a root of positive valuation if and only if the Gauss valuation of $R$ is strictly less than the 
valuation of the constant term of $R$:
\begin{equation}
\label{eq-R-val}
\val(R)<\val(R(0))\, .
%\qquad
%\val(R(Z))=\min_{0\leq i\leq D_R}\val(c_i(\varepsilon)).
\end{equation}
Indeed, write $R(Z)=c_{D_R}\prod_{i=1}^{D_R}(Z-z_i^*)$ over $\KK$.
Multiplicativity gives
\[
\val(R)=\val(c_{D_R})+
\sum_{i=1}^{D_R}\min\{0,\val(z_i^*)\}.
\]
If $R(0)=0$, then $0$ is a root of valuation $\infty$, and
\eqref{eq-R-val} holds since $R\neq0$. Otherwise all roots are nonzero,
and
\begin{align*}
\val(R(0))-\val(R)
&=\sum_{i=1}^{D_R}
\bigl(\val(z_i^*)-\min\{0,\val(z_i^*)\}\bigr)\\
&=\sum_{i=1}^{D_R}\max\{0,\val(z_i^*)\}.
\end{align*}
This proves the  relation given in~\eqref{eq-R-val}. 

It remains to show how the condition~\eqref{eq-R-val} can be seen as an instance of the MRV problem.
The issue is that the polynomial $R$ is not directly defined as a resultant.  Rather, $R$ 
was defined in~\eqref{eq:iterated} as an iterated trailing coefficient of a resultant.
To address this,
 let $D_0$ bound the degrees of $R_0$ in $T,W,$ and $\varepsilon$, put $m:=D_0+1$, and define
\[
\widehat R_0(\varepsilon,Z)
:=R_0(\varepsilon^{m^2},\varepsilon^m,\varepsilon,Z).
\]
Using the fact that the exponent $am^2+bm+j$ of $\varepsilon$ uniquely encodes the triple $(a,b,j)$ in lexicographic order, Lemma~\ref{lem-weighted-valuation} proves the equivalence
\[
\val(R_2(\varepsilon,Z))<\val(R_2(\varepsilon,0))
\quad\Longleftrightarrow\quad
\val(\widehat R_0(\varepsilon,Z))<\val(\widehat R_0(\varepsilon,0)).
\]
Since $\widehat R_0$ is  the resultant of the forms defining $R_0$ after the substitutions $T=\varepsilon^{m^2}$ and $W=\varepsilon^m$, the right-hand side is an MRV query, with distinguished monomial $Z^0$. This completes  the reduction in Theorem~\ref{theo-main-reduc}.

\paragraph{The MRV problem in the Counting Hierarchy.}
In Section~\ref{sec-algorithm} we show that the MRV problem lies in the counting hierarchy.  The key tool is~\cite[Theorem~6.7]{AGS26}, which implies that determining whether a specified coefficient of a resultant vanishes is in $\CH$.

To see how this can be used to solve an instance of MRV, write the resultant as
\[
 R(\varepsilon,\bm Y)
=\sum_{\bm\beta,j}c_{\bm\beta,j}\varepsilon^j\bm Y^{\bm\beta}.
\]
For the distinguished monomial $\bm Y^{\bm\beta_0}$, the strict inequality in the definition of MRV can be written
\begin{gather}
\exists\bm\beta\;\exists j:
c_{\bm\beta,j}\neq0
\;\text{and}\;
c_{\bm\beta_0,j'}=0\; \text{for every }0\leq j'\leq j.
\label{eq:cond}
\end{gather}
Indeed, the nonzero coefficient $c_{\bm\beta,j}$ witnesses that the minimum valuation of the resultant is at most $j$, while the universal vanishing condition on $c_{\bm\beta_0,j'}$ says that the distinguished coefficient has valuation strictly greater than $j$. Standard resultant degree bounds give singly exponential bounds on the relevant exponents, so all indices have polynomial-length binary encodings.  Since the MRV problem involves a fixed quantifier alternation over coefficient tests, it also lies in $\CH$. This proves~$\mathrm{MRV}_k\in\CH$.

\section{Encoding and Sampling Conventions}
\label{sec-prelim}
In this section we discuss the representation of field elements and polynomials, as well as the random-sampling conventions used in the complexity bounds below.

\paragraph{Field elements.}
We denote elements of $\mathbb Q$ by pairs of
binary integers. A finite field $\F_q$, where $q=p^e$, is represented as
$\F_p[U]/(\mu(U))$, with $p$ prime and $\mu$ a specified monic irreducible
polynomial of degree $e$. Writing $u$ for the residue class of $U$,
a field element is represented in the power basis $1,u,\ldots,u^{e-1}$. 
We denote the total input length by $N$. Thus the number of
variables and monomials, the coefficient bit lengths, and the logarithms
of the degrees are polynomially bounded in $N$. Over $\Q$, clearing the
denominators of each input polynomial does not change its APS status.
We also assume that all polynomials in an instance of $\APS_k$ are non-constant:
identically zero polynomials may be deleted without changing the APS status, 
while a nonzero constant polynomial makes
the instance negative. 

\paragraph{Polynomials.}
Unless otherwise stated, polynomials are represented sparsely, as lists of monomials
with exponents written in binary. In particular, this is the
representation used for instances of the MRV problem.  

The following proposition shows that the APS problem is robust with respect to the encoding of polynomials:
there is a polynomial-time reduction of the version of the problem in which polynomials are encoded 
as arithmetic circuits to the version with the so-called dense encoding, where 
polynomials are represented as arrays of coefficients.  Thus the versions with circuit, sparse, and dense encodings
of the input
are all polynomial-time equivalent.

\begin{proposition}
\label{prop:dense-representations}
Let \(k\) be either \(\mathbb Q\) or a finite field. 
The version of \(\APS_k\) in which the input polynomials are represented by
division-free arithmetic circuits can be reduced in polynomial time to the
version with a dense encoding of the input polynomials.
\end{proposition}
\begin{proof}
Let \(C_1,\ldots,C_r\) be division-free arithmetic circuits in variables
\(\bm W=(W_1,\ldots,W_t)\). We may assume that every internal gate has
fan-in two and is labelled by \(+\) or \(\times\). By adding
trivial gates, if necessary, we assume also that the output gate is not a leaf.

Introduce a variable \(Z_g\) for every internal gate \(g\). For a gate
\(g\) with children \(a,b\), define
\[
R_g:=
\begin{cases}
Z_g-Z_a-Z_b, &\text{if }g=a+b,\\
%Z_g-Z_a+Z_b, &\text{if }g=a-b,\\
Z_g-Z_aZ_b,  &\text{if }g=ab.
\end{cases}
\]
In the above equations, if either $a$ or $b$ is a leaf then 
$Z_a$ and $Z_b$ stand for the corresponding input variable 
from $\bm W$ or a field element.  Let
\(\bm X=(\bm W,\bm Z)=(X_1,\ldots,X_m)\) denote the tuple of
all input and gate variables.

For a gate \(g\), let \(\mu(g)\) be the maximum number of multiplication gates on a directed path starting from a leaf and ending at \(g\), and let
\(L:=\max\{1,\max_g\mu(g)\}\). Put
\(h:=\lceil\log_2(L+1)\rceil\), so that \(2^h>L\).
For every \(j\in[m]\), introduce auxiliary variables
\(Y_{j,1},\ldots,Y_{j,h}\), and consider the following system:
\begin{align}
Y_{j,1}-X_j^2&=0
&& (j\in[m]), \label{eq:circuit-amp-1}\\
Y_{j,\ell+1}-Y_{j,\ell}^2&=0
&& (j\in[m],\ 1\leq \ell<h), \label{eq:circuit-amp-2}\\
R_g(\bm X)&=0
&& (g\text{ an internal gate}), \label{eq:circuit-amp-3}\\
Y_{j,h}R_g(\bm X)&=0
&& (j\in[m],\ g\text{ an internal gate}), \label{eq:circuit-amp-4}\\
Z_{o_i}&=0
&& (i\in[r]), \label{eq:circuit-output}
\end{align}
where \(o_i\) is the output gate of \(C_i\). We will show that this is an
equivalent instance of \(\APS_k\).

Suppose that the original circuit instance has an approximate zero
\(\bm w\in\KK^t\). Assign each gate variable \(Z_g\) the 
value computed by gate \(g\) at \(\bm w\). Then
equation~\eqref{eq:circuit-amp-3} holds exactly. 
Choose values of the auxiliary variables so that
the equations
\eqref{eq:circuit-amp-1}--\eqref{eq:circuit-amp-2} also hold exactly.
Since \(Z_{o_i}=C_i(\bm w)\) has positive valuation for every \(i\),
the constructed system has an approximate zero.

Conversely, let \((\bm x,\bm y)\) be an approximate zero  of
the system
\eqref{eq:circuit-amp-1}--\eqref{eq:circuit-output}, and put
\[q:=-\min(0,\val(x_1),\ldots,\val(x_m)).\] 
We show that the original circuit instance of APS is approximately satisfiable.   
To this end,
we first claim that
\(\val(R_g(\bm x))>Lq\) for every gate \(g\).

If \(q=0\), the claim follows immediately from
\eqref{eq:circuit-amp-3}. Suppose therefore that \(q>0\), and choose
\(j_0\in[m]\) such that \(\val(x_{j_0})=-q\). Approximate satisfaction of
the squaring equations gives, by induction,
\(\val(y_{j_0,\ell})=-2^\ell q\) for \(1\leq\ell\leq h\).
%Indeed, the assertion holds for \(\ell=1\) because
%\(\val(x_{j_0}^2)=-2q<0\), while
%\(\val(y_{j_0,1}-x_{j_0}^2)>0\); the inductive step is identical.
In particular, \(\val(y_{j_0,h})=-2^hq\). Hence
\eqref{eq:circuit-amp-4} gives
\[0<\val\bigl(y_{j_0,h}R_g(\bm x)\bigr)=-2^hq+\val(R_g(\bm x)),
\]
and therefore \(\val(R_g(\bm x))>2^hq>Lq\).  This establishes the claim.

Let \(c_g\)
denote the value computed by gate \(g\) when the original circuit is
evaluated at the \(\bm W\)-coordinates of \(\bm x\). Define \(z_g\) when \(g\) is a leaf by the value of the corresponding input variable or field constant. For every gate put \(e_g:=z_g-c_g\). Notice that \(e_g:=0\) for input and constant gates. We
claim by induction over the circuit that
\[
\val(e_g)>(L-\mu(g))q
\]
for every gate \(g\).

For a leaf gate \(g\), the result is trivial since \(e_g:=0\).

For an addition gate \(g\) with children \(a,b\), we have
\begin{align*}
 e_g = z_g - c_g = R_g + z_a + z_b - c_z - c_b = R_g+e_a+ e_b,
\end{align*}
and the claim follows immediately from the
induction hypothesis.

Now suppose that \(g=ab\). 
% Write \(z_a,z_b\) for the values of the gate variables
% corresponding to $a$ or $b$, or, in case $a$ or $b$ is a leaf, the value
% of the corresponding input variable or field constant.  
Then
\[e_g= z_g - c_g = R_g + z_az_b-c_ac_b = R_g + z_ae_b + z_be_a-e_ae_b = R_g+z_ae_b+z_be_a-e_ae_b.\]
Both \(z_a\) and \(z_b\) have valuation at least \(-q\) by the
definition of \(q\).
% for input and internal-gate variables, while field
% constants have nonnegative valuation. 
Thus
\[
\val(z_ae_b)>
(L-\mu(b)-1)q
\geq (L-\mu(g))q,
\]
and similarly for \(z_be_a\), since
\(\mu(g)=1+\max(\mu(a),\mu(b))\). Moreover,
\[
\val(e_ae_b)>
(2L-\mu(a)-\mu(b))q
>
(L-\mu(g))q.
\]
Together with \(\val(R_g)>Lq\), this proves the claim.

For the output gate \(o_i\) of $C_i$, we have \(\mu(o_i)\leq L\), and hence
\(\val(e_{o_i})>0\). On the other hand,
\eqref{eq:circuit-output} gives \(\val(z_{o_i})>0\). Therefore
\(C_i(\bm w)=c_{o_i}=z_{o_i}-e_{o_i}\) has positive valuation for every
\(i\), where \(\bm w\) denotes the \(\bm W\)-coordinates of \(\bm x\).
Thus \(\bm w\) is an approximate zero of the original circuit instance.

It remains to observe that the reduction can be carried out in polynomial
time using a dense encoding for the output polynomials. 
The reduction introduces one new variable for each circuit gate and $mh$ auxiliary variables, 
which makes polynomially many new variables in the 
circuit size.   The polynomials in the constructed system have degree at most 3
and so they can be written explicitly in dense representation
in polynomial time.
\end{proof}

\paragraph{Sampling.}
We recall the Schwartz-Zippel lemma: for a non-zero polynomial 
$Q\in k[Z_1,\ldots,Z_m]$ and point~$\bm a$ sampled uniformly at random from a given set $S\subseteq k$,
we have
\[ \Pr_{\bm a \in S^m} \left[ Q(\bm a)=0 \right]\leq \frac{\deg(Q)}{|S|}\, \, .
\]

When a random choice must avoid a nonzero polynomial of degree at most
$B=2^{\poly(N)}$, we sample from a finite set $S\subseteq\ell\subseteq K$
of size a fixed large multiple of $B$. For $k=\Q$, we take $\ell=\Q$ and an integer
grid. For $k=\F_q$, we pass when necessary to an extension
$\ell=\F_{q^r}\subseteq K$, where $r=\poly(N)$ and $q^r$ is sufficiently
large. Such an extension, together with an explicit embedding of $k$, can
be constructed in randomised polynomial time; see
\cite[Lemma~2.13]{AGS26}. A polynomial-time cutoff for the Las Vegas
construction contributes an arbitrarily small constant to the error
probability. Sampling can likewise be implemented with bounded running
time and an arbitrarily small failure probability. 
%\textcolor{blue}{Mahsa: Improve the following sentences:} Alternatively, after
%constructing $\ell$, one can sample exactly from a power-of-two-sized
%subset of its standard encoding. The Schwartz--Zippel bound applies over
%$\ell$ just as it does over $\Q$. In particular, an integer grid reduced
%modulo $p$ must not be used as a substitute for a sufficiently large finite
%field. Whenever intermediate instances have coefficients in $\ell$, their
%solutions and image closures are still taken over the same $K$.

\section{Reducing APS to MRV}
\label{sec-algebaric}
The goal of this section is to prove Theorem~\ref{theo-main-reduc}.

\propsubsetsnindependent*
\begin{proof}
Let $\bm f=(f_1,\ldots,f_m)\in k[X_1,\ldots,X_s]^m$. Write $\bm X=(X_1,\ldots,X_s)$.
Introduce a new indeterminate $Y$ and define
\[
\bm g=(g_1,\ldots,g_{m+s+1})
     :=(f_1,\ldots,f_m,Y,YX_1,\ldots,YX_s).
\]
The transcendence degree of $\bm g$ is $s+1$: the fraction field of
$k[Y,YX_1,\ldots,YX_s]$ is $k(Y,X_1,\ldots,X_s)$.
Moreover, $\bm f$ and $\bm g$ have the same APS status. One direction
follows by projecting onto the first $m$ coordinates. Conversely, given
an approximate solution $\bm x(\varepsilon)\in\KK^s$ for $\bm f$, choose an integer
$r\geq1$ such that $r+\val(x_j(\varepsilon))>0$ whenever $x_j(\varepsilon)\neq0$.
Then $(\bm x(\varepsilon),\varepsilon^r)$ is an approximate solution for $\bm g$.
This uses only valuations, not a series representation of the coordinates.

Let $\bm a=(a_{ij})$ be an $(s+2)\times(m+s+1)$ matrix with entries in
{a field $\ell$ with $k\subseteq\ell\subseteq K$}, and set
\[
h_i^{\bm a}(\bm X,Y)
   :=\sum_{j=1}^{m+s+1}a_{ij}g_j(\bm X,Y),
\qquad i\in[s+2].
\]
We show that a random choice of $\bm a$ makes
$\bm h^{\bm a}=(h_1^{\bm a},\ldots,h_{s+2}^{\bm a})$ minimally dependent
and preserves APS with high probability.

Order the variables as $(Y,X_1,\ldots,X_s)$ and put
$\bm p=(1,0,\ldots,0)$. Write $\bm b:=J_{\bm g}(\bm p)$ for the
$(m+s+1)\times(s+1)$ Jacobian matrix at this point. Its last $s+1$ rows
form the matrix
\[
J_{(Y,YX_1,\ldots,YX_s)}(\bm p)=I_{s+1},
\]
so $\bm b$ has rank $s+1$ in every characteristic.

Let $\bm A=(A_{ij})$ now denote a matrix of indeterminates. For an
$(s+1)$-element subset $I\subseteq[s+2]$, let $\bm A_I$ be its row
submatrix and define
\[
D_I(\bm A):=\det(\bm A_I\bm b).
\]
This is a nonzero polynomial of degree at most $s+1$: specialising the
rows of $\bm A_I$ to select $Y,YX_1,\ldots,YX_s$ makes the determinant
$1$. Consequently,
\[
D(\bm A):=\prod_{\substack{I\subseteq[s+2]\\|I|=s+1}}D_I(\bm A)
\]
is nonzero and has degree at most $(s+1)(s+2)$. If $D(\bm a)\neq0$,
the chain rule gives
$J_{\bm h_I^{\bm a}}(\bm p)=\bm a_I\bm b$, an invertible matrix,
for every such $I$.

Only the sufficient direction of the Jacobian criterion is needed
here. To spell this out in positive characteristic, suppose that a tuple
with invertible Jacobian satisfies a nonzero relation
$P\in\ell[Z_1,\ldots,Z_{s+1}]$ of least total degree. Differentiating the
relation and using the invertible Jacobian over the rational-function
field shows that all the partial derivatives of $P$ vanish after
substitution. Minimality implies that all these partial derivatives
are identically zero. In characteristic zero this makes $P$ constant,
a contradiction. In characteristic $p$, all its monomial exponents
are divisible by $p$; since the finite field $\ell$ is perfect,
{$P=Q^p$ for a polynomial $Q\in\ell[Z_1,\ldots,Z_{s+1}]$ of smaller degree.}
The polynomial ring is a domain, so $Q$ is also a relation, again a
contradiction. Thus every $(s+1)$-subtuple of $\bm h^{\bm a}$ is
algebraically independent. The entire tuple is dependent because it
has $s+2$ coordinates in $s+1$ variables. {It is therefore minimally dependent.}
No converse Jacobian criterion is used.

Let
$d\geq\max\{2,\deg f_1,\ldots,\deg f_m\}$ and let $S\subseteq\ell$
be finite. Choose the entries of $\bm a$ independently and uniformly
from $S$. Schwartz--Zippel gives
\[
\Pr[D(\bm a)=0]\leq\frac{(s+1)(s+2)}{|S|}.
\]

It remains to preserve APS. 
The random-projection theorem of Guo, Saxena and Sinhababu,
which holds over arbitrary fields, gives
\cite[Theorem~4.6]{GSS19}
\[
\Pr\bigl[{\text{$\bm g$ is a positive APS instance}}\ \Longleftrightarrow\
                 {\text{$\bm h^{\bm a}$ is a positive APS instance}}\bigr]
\geq1-\frac{(s+2)d^{s+1}}{|S|}.
\]
Here we apply their theorem over $K$, with sampling set $S\subseteq K$.
By the union bound, the probability that {$\bm h^{\bm a}$ is minimally dependent
and has} the same APS status as $\bm f$ is at least
\[
1-\frac{(s+2)(s+1+d^{s+1})}{|S|}.
\]
For any fixed error tolerance $\eta>0$, take
$|S|\geq\eta^{-1}(s+2)(s+1+d^{s+1})$.
Over $\Q$ we may use an integer grid of power-of-two cardinality.
Over a finite field, use a sufficiently large finite extension and a
set of distinct field elements as described in Section~\ref{sec-prelim}.
The extension degree and the binary length of $|S|$ are polynomial in
the input length. The lists of monomials of $\bm g$ and
$\bm h^{\bm a}$ also have polynomial size. Taking, for example,
$\eta=1/12$, and budgeting the field-construction and sampling error
separately, proves the proposition with room for the later random choices.
\end{proof}

\propchoiceaipositiveimage*
% \color{blue}
% Seb: I suggest this modification for the proof of the proposition. 
% The point is that we do not need nonemptiness in the statement of Claim 8, and so keep the proof there very short. On the other way, we make Claim 7 a bit more general to ensure that $H(a_1\varepsilon,\ldots,a_n\varepsilon,Z)$ has a root even in the negative APS instance. But we also define by the way a G for the negative APS instance case, and so make the final argument more symmetric. At the end, nonemptiness of V follows from Claim 7 and 9.
% I was hoping to get a `nicer' \(G\), defined similarly based on whether it is an APS instance or not. But I did not succeed. It does not really affect the proof, it was just to try to make it a bit cleaner.
% ---------------------------------------------------------------------
\begin{proof}
We divide the proof into four claims and conclude with a brief argument
that the proposition follows from them.

\begin{claim}
\label{claim:ONE}
There is an irreducible polynomial $H\in k[Y_1,\ldots,Y_{n+1}]$ of
degree at most $\Delta^{n+1}$ such that
\[
k[f_1,\ldots,f_{n+1}]
  \simeq k[Y_1,\ldots,Y_{n+1}]/(H).
\]
The same presentation holds after extending scalars to $K$.
Moreover, {$\bm f$ is a positive APS instance} if and only if $H(\bm0)=0$.
\end{claim}
\begin{proof}
Let $I$ be the kernel of the homomorphism
$k[\bm Y]\to k[\bm X]$ that sends $Y_i$ to $f_i$.
It is a nonzero prime ideal, since the $f_i$'s are dependent and
$k[\bm X]$ is an integral domain. {By minimal dependence, their}
transcendence degree is $n$, so $I$ has height one. The polynomial
ring $k[\bm Y]$ is a unique factorisation domain, and hence $I=(H)$
for an irreducible polynomial $H\in k[\bm Y]$. In the rational case
only, we may normalise $H$ to be a primitive polynomial over $\Z$.

Perron's degree bound on annihilating polynomials~\cite[Theorem~1.1]{Ploski2005}
gives a nonzero element of $I$ of degree at
most $\prod_{i=1}^{n+1}\deg f_i$. Since $H$ divides this element,
\[
{D_H:=\deg H}\leq\prod_{i=1}^{n+1}\deg f_i\leq\Delta^{n+1}.
\]
Extension of scalars from $k$ to $K$ is exact. Thus the kernel of
$K[\bm Y]\to K[\bm X]$ is still generated by $H$, and
$K[\bm Y]/(H)\simeq K[f_1,\ldots,f_{n+1}]$.
In particular, this quotient is a domain. We obtain
\[
\overline{\bm f(K^n)}
   =\{\bm y\in K^{n+1}:H(\bm y)=0\}.
\]
The origin belongs to this set exactly when $H(\bm0)=0$.
\end{proof}

For the next claim, write $D_H:=\deg H \ge 1$ (since $H$ is not constant).
\begin{claim}
\label{claim:TWO}
After permuting its coordinates, there is a nonzero homogeneous polynomial $G\in k[Y_1,\ldots,Y_n]$ of {degree at most $D_H-1$} such that, 
for all $\bm a\in K^n$ with
$G(\bm a)\neq0$, the polynomial
$H(a_1\varepsilon,\ldots,a_n\varepsilon,Z)$ has a root in $\KK$. Moreover, if {$\bm f$ is a positive APS instance}, we can even find a root with positive valuation.
\end{claim}
\begin{proof}
Write $H=H_v+{H_{v+1}+\cdots+H_{D_H}}$, where $H_i$ is homogeneous of
degree $i$ and $H_{v}\neq0$. We saw that \(v = 0\) if and only if {$\bm f$ is a negative APS instance}. Let us define
\begin{align*}
    d = \begin{cases}
        v & \textrm{if }v \ne 0 \\
        D_H & \textrm{otherwise}.
    \end{cases}
\end{align*}
In all cases, \(d \ge 1\), so we can assume, up to permuting the coordinates, that $Y_{n+1}$ occurs in $H_{d}$. Choose a positive exponent $r$ for which
\[
G(Y_1,\ldots,Y_n):=\coeff_{Y_{n+1}^r}(H_{d})\neq0.
\]
Then $G$ is homogeneous of degree ${d-r\leq D_H-1}$. The requirement
$r>0$ is important: the lowest power of $Y_{n+1}$ occurring in
$H_{d}$ could be zero, even though $Y_{n+1}$ occurs in $H_{d}$.

Fix $\bm a\in K^n$ with $G(\bm a)\neq0$. Consider the polynomial \(H(a_1\varepsilon,\ldots,a_n \varepsilon,Z) \in \left(K[\varepsilon]\right)[Z]\). It has positive \(Z\)-degree since the coefficient of the monomial \(Z^r\varepsilon^{d-r}\) is \(G(\bm a)\) which is nonzero. Hence \(H(a_1\varepsilon,\ldots,a_n \varepsilon,Z) \) has a root in \(\mathcal{K}\), which proves the first part of the claim.

Assume now that {$\bm f$ is a positive APS instance}. Fix again $\bm a\in K^n$ with $G(\bm a)\neq0$. For a fresh indeterminate
$U$, define
\begin{align*}
{H^{\circ}(\varepsilon,U)}&:=\varepsilon^{-d}H\left(a_1 \varepsilon,\ldots, a_n \varepsilon,\varepsilon U\right)\\
      &={\sum_{\ell=d}^{D_H}} \varepsilon^{\ell-d}H_\ell(a_1,\ldots,a_n,U)
        \in K[\varepsilon,U].
\end{align*}
Its reduction ${H^{\circ}(0,U)}=H_{d}(\bm a,U)$ is nonconstant, since its
coefficient of $U^r$ is $G(\bm a)\neq0$.

We claim that ${H^{\circ}(\varepsilon,U)}$ has a root in $\KK$ of nonnegative valuation.
Suppose, to the contrary, that every root $u_1,\ldots,u_m$, listed
with multiplicity, has negative valuation. In particular, none is zero,
and
\[
{H^{\circ}(\varepsilon,U)}={H^{\circ}(\varepsilon,0)}\prod_{\ell=1}^m\left(1-\frac{U}{u_\ell}\right).
\]
Since $\val({H^{\circ}(\varepsilon,0)})\geq0$ and $\val(1/u_\ell)>0$, {the coefficient of each positive power of $U$ has strictly positive valuation}. The {reduction of $H^{\circ}$} would
therefore be constant, contradicting the preceding paragraph. Thus
there is a root $u^*$ with $\val(u^*)\geq0$. Setting {$z^*:=\varepsilon u^*$} gives
\begin{equation}
H(a_1\varepsilon,\ldots,a_n\varepsilon,z^*)=0,
\qquad \val(z^*)\geq1.
\label{eq:pos-root}
\end{equation}
\end{proof}

For each $j\in[n]$, apply Perron's bound to $f_1,\ldots,f_n,X_j$.
It gives a nonzero polynomial
\[
P_j(Y_1,\ldots,Y_n,U)
=q_j(Y_1,\ldots,Y_n)U^{m_j}
 +\sum_{\ell=0}^{m_j-1} %(\bm Y)U^\ell \in k[\bm Y,U]}
 {q_{j,\ell}(Y_1,\ldots,Y_n)U^\ell\in k[Y_1,\ldots,Y_n,U]}
\]
such that
\[
P_j(f_1,\ldots,f_n,X_j)=0,
\qquad
\deg P_j\leq\prod_{i=1}^n\deg f_i\leq\Delta^n.
\]
The degree $m_j$ in $U$ is positive, since otherwise $P_j$ would be
a nonzero relation among {the algebraically independent polynomials $f_1,\ldots,f_n$}.
Take $q_j\neq0$ and define
\begin{gather}
h:=q_1\cdots q_n\neq0,
\qquad \deg h\leq n\Delta^n.
\label{def:H}
\end{gather}

\begin{claim}
\label{claim:THREE}
If $\bm a\in K^n$ satisfies $h(\bm a)\neq0$, then
\begin{equation}
V:=\{\bm x\in\KK^n:f_i(\bm x)=a_i\varepsilon \text{ for every }i\in[n]\}
\label{eq:VAR}
\end{equation}
is finite, whether or not {$\bm f$ is a positive APS instance}.
\end{claim}
\begin{proof}
For each $j$, the polynomial $q_j(a_1\varepsilon,\ldots,a_n\varepsilon)$ is nonzero:
its value at $\varepsilon=1$ is $q_j(\bm a)\neq0$. Thus
$P_j(a_1\varepsilon,\ldots,a_n\varepsilon,U)$ has degree $m_j$ in $U$.
If $\bm x\in V$, then
\[
P_j(a_1\varepsilon,\ldots,a_n\varepsilon,x_j)=P_j(f_1(\bm x),\ldots,f_n(\bm x),x_j)=0.
\]
Each coordinate has at most $m_j$ possible values, so
$|V|\leq\prod_{j=1}^n m_j$.
\end{proof}

\begin{claim}
\label{claim:FOUR}
If $\bm a\in K^n$ satisfies $h(\bm a)\neq0$, then every root of
$H(a_1\varepsilon,\ldots,a_n\varepsilon,Z)$ in $\KK$ lies in $f_{n+1}(V)$.
\end{claim}
\begin{proof}
Define rings $B\subseteq C$ by
\begin{equation}
B:=K[f_1,\ldots,f_n,f_{n+1}],
\qquad C:=K[X_1,\ldots,X_n].
\label{eq:RINGS}
\end{equation}
{Since $f_1,\ldots,f_n$ are algebraically independent,}
$h(f_1,\ldots,f_n)\neq0$. Write
\[
B_h:=B\left[\frac{1}{h(f_1,\ldots,f_n)}\right],
\qquad
C_h:=C\left[\frac{1}{h(f_1,\ldots,f_n)}\right].
\]
Every $q_j(f_1,\ldots,f_n)$ is invertible in $B_h$. Dividing the
relation $P_j(f_1,\ldots,f_n,X_j)=0$ by its leading coefficient gives
a monic equation for $X_j$ over $B_h$. Therefore $C_h$ is integral
over $B_h$.

Let $z^*\in\KK$ be any root of $H(\bm a\varepsilon,Z)$. By
Claim~\ref{claim:ONE}, the presentation
{$B\simeq K[Y_1,\ldots,Y_n,Z]/(H)$, with $Y_{n+1}$ renamed $Z$,} gives a $K$-algebra homomorphism
$\varphi:B\to\KK$ with
\[
\varphi(f_i)=a_i\varepsilon\quad(\text{with }i\in[n]),
\qquad \varphi(f_{n+1})=z^*.
\]
As $\varphi(h(f_1,\ldots,f_n))=h(\bm a\varepsilon)\neq0$, it extends to
$\varphi':B_h\to\KK$. Proposition~\ref{prop-lang} gives an extension
$\varphi'':C_h\to\KK$. Define
\[
\bm x^*:=\bigl(\varphi''(X_1),\ldots,\varphi''(X_n)\bigr)\in\KK^n.
\]
Then $f_i(\bm x^*)=a_i\varepsilon$ for $i\in[n]$, and
$f_{n+1}(\bm x^*)=z^*$. Hence $z^*\in f_{n+1}(V)$.
\end{proof}

We now conclude the proof of the proposition. Use the permutation from Claim~\ref{claim:TWO} and set {$Q_{\mathrm{bad}}:=Gh$}. We know 

% If {$\bm f$ is a positive APS instance},
% If {$\bm f$ is a negative APS instance}, take any coordinate order and set {$Q_{\mathrm{bad}}:=h$}.
% In either case, {$Q_{\mathrm{bad}}\neq0$} and
\[
\deg Q_{\mathrm{bad}} \leq \Delta^{n+1}-1+n\Delta^n=\Delta_{\mathrm{bad}}-1.
\]
If ${Q_{\mathrm{bad}}(\bm a)} \neq 0$, Claim~\ref{claim:THREE} makes $V$ finite. Moreover, Claims~\ref{claim:TWO}
and~\ref{claim:FOUR} make \(V\) nonempty. In the positive case, Claims~\ref{claim:TWO}
and~\ref{claim:FOUR} give even a point of $V$ whose last image coordinate
has positive valuation. Conversely, any point in \(V\) whose last image coordinate
has positive valuation for some \(\bm a\) in \(K^n\) is an approximate solution of \(\bm f\), and hence
ensures that \(\bm f\) is a positive APS instance.
\end{proof}

Let \(\bm a \in K^n\) be such that \(Q_{\mathrm{bad}}(\bm a) \ne 0\) and the tuple \(\bm f\) ordered as described in Proposition~\ref{prop-choice-ai-positive-image}.
Denote by \(G_i(X_0,\ldots,X_n)\) the homogeneisation of \(f_i-a_i\varepsilon\) (for \(1 \le i \le n\)) and by \(F(X_0,\ldots,X_n)\) the homogeneisation of \(f_{n+1}\).  Recall that we defined 
\begin{align*}
& \widetilde G_i:=G_i+TX_i^{d_i}\qquad\quad \textrm{with } d_i:=\deg_{X_0,\ldots,X_n}G_i, \\
& L:=ZX_0^d+\sum_{i=1}^nW^iX_i^d-F
\end{align*}
where \(T\), \(W\), and \(Z\) are new indeterminates. We defined also
\begin{equation}
    R_0(T,W,\varepsilon,Z):=\Res(\widetilde G_1,\ldots,\widetilde G_n,L).
    \label{eq:R0-second}
\end{equation}

\begin{lemma}
\label{lem-resultant-images}
Suppose that the set $V$ in~\eqref{eq:VAR} is finite and nonempty.
Then the resultant $R_0$ in~\eqref{eq:R0-second} is nonzero and
\[
R_2(\varepsilon,Z):=\TT_W\TT_T R_0(T,W,\varepsilon,Z)
 =c\prod_{\bm x\in V}\bigl(Z-f_{n+1}(\bm x)\bigr)^{m_{\bm x}},
\]
where $c\in\KK^\times$ and every $m_{\bm x}$ is a positive integer.
\end{lemma}
\begin{proof}
We use the trailing-resultant statements of \cite[Lemmas~4.6 and~4.7]{AGS26}, which follow~\cite{Ierardi89} and apply over any field. These lemmas concern a certain limit set $V_0^{*} \subseteq \mathbb P^n(\KK)$ of projective points.
Here we do not need to know how $V_0^{*}$ is defined (see Section 2 in~\cite{Ierardi89} or Section 4.2 in~\cite{IerPhD} for details).  The properties of $V_0^{*}$ that we use, which are stated in
Lemma 4.6 in~\cite{AGS26}, are that it is a finite subset of the
projective variety \(V_{\mathbb P^n(\mathcal{K})}(G_1,\ldots,G_n)\) that contains all of its isolated points.

The affine part of \(V_{\mathbb P^n(\mathcal{K})}(G_1,\ldots,G_n)\) is exactly 
\begin{align*}
    \{ [1 :\bm x] \mid \bm x \in V\} 
\end{align*}
and, since \(V\) is finite, all the above points are isolated. It follows that the affine part of $V_0^{*}$ equals the affine part of \(V_{\mathbb P^n(\mathcal{K})}(G_1,\ldots,G_n)\):
\begin{equation*}
    \{x \mid [1:\bm x] \in V_0^*\} = V.
\end{equation*}

Applying Lemma 4.7 in~\cite{{AGS26}} with $m=1$ and the form \(L\) as above,
shows that the polynomial $\TT_T R_0 \in K[W,\varepsilon,Z]$ factors into a product of linear forms over \(\KK\). Furthermore, for each projective point \([x_0 : \cdots : x_n] \in V_0^*\), the corresponding linear factor is
\begin{equation*}
    Zx_0^d + \sum_{i=1}^n W^i x_i^d - F([x_0 : \cdots : x_n])
\end{equation*}
and every factor of \(\TT_T R_0\) has this structure for some point \([x_0 : \cdots : x_n]\) in \(V_0^*\).

At an affine limit point, normalised to $[1:\bm x] \in V_0^*$, the evaluation of this factor is
\begin{equation*}
Z+\sum_{i=1}^n W^ix_i^d-f_{n+1}(\bm x),
\end{equation*}
a nonzero polynomial whose trailing coefficient in $W$ is
$Z-f_{n+1}(\bm x)$. At a limit point at infinity $[0:\bm x] \in V_0^*$,
its evaluation is
\begin{equation*}
    \sum_{i=1}^n W^ix_i^d-F(0,x_1,\ldots,x_n).
\end{equation*}
This polynomial is nonzero: indeed some $x_i$ is non-zero, and the positive powers
$W^i$ are distinct. Its trailing coefficient in $W$ does not depend on \(Z\), and so  is a nonzero scalar in \(\mathcal{K}\).
Taking the trailing coefficient in $W$ of the polynomial \(\TT_T R_0\) 
now gives the claimed formula because the trailing-coefficient is multiplicative.
If a projective limit point occurs with multiplicity \(m\), its linear factor occurs with the same multiplicity \(m\); taking \(\TT_W\) preserves these multiplicities because trailing coefficients multiply.
\end{proof}

\begin{lemma}
\label{lem-weighted-valuation}
Let $P\in k[T,W,\varepsilon,Z]\setminus\{0\}$ have degree at most $D_0$ in
each of $T,W,\varepsilon$. Set $m:=D_0+1$ and
\[
P_2(\varepsilon,Z):=\TT_W\TT_T P,
\qquad 
\widehat P(\varepsilon,Z):=P(\varepsilon^{m^2},\varepsilon^m,\varepsilon,Z).
\]
Then $\widehat P\neq0$, and
\[
\val(P_2(\varepsilon,Z))<\val(P_2(\varepsilon,0))
\quad\Longleftrightarrow\quad
\val(\widehat P(\varepsilon,Z))<\val(\widehat P(\varepsilon,0)).
\]
\end{lemma}
\begin{proof}
A monomial $T^aW^b\varepsilon^jZ^i$ in $P$ becomes $\varepsilon^{am^2+bm+j}Z^i$ in $\widehat P$.
Since $0\leq a,b,j<m$, the exponent uniquely encodes $(a,b,j)$,
and its numerical order is the lexicographic order on those triples.
There is therefore no collision of distinct monomials and no
cancellation caused by the substitution, in any characteristic.

Let $(a_0,b_0)$ be the lexicographically least pair 
such that the coefficient of $T^{a_0}W^{b_0}$ in $P$ is non-zero.
Then, for any \(0 \le i \le D_0\),
\begin{equation}\label{eq:TT_val_coeff}
    \coeff_{\varepsilon^i} P_2 = \coeff_{T^{a_0}W^{b_0}\varepsilon^i} P 
    = \coeff_{\varepsilon^{a_0m^2 + b_0m+i}} \widehat{P}.
\end{equation}
Let \(c_0 := \val (P_2(\varepsilon,Z))\). Since \(P_2\) is a nonzero polynomial of degree at most \(D_0\), we know that \(c_0 \le D_0\). By definition
\begin{equation}\label{eq:TT_val_val}
    c_0 = \val(\coeff_{T^{a_0}W^{b_0}} P) = \val(\widehat{P}) - a_0 m^2 - b_0 m.
\end{equation}
We show the equivalence of the lemma
\begin{align*}
    \val(P_2(\varepsilon,Z))<\val(P_2(\varepsilon,0))
    & \quad\Longleftrightarrow\quad
    Z \textrm{ divides } \coeff_{\varepsilon^{c_0}} P_2 \\
    %& \quad\Longleftrightarrow\quad
    %Z \textrm{ divides } \coeff_{T^{a_0}W^{b_0}\varepsilon^{c_0}} P \\
    & \quad\Longleftrightarrow\quad
    Z \textrm{ divides } \coeff_{\varepsilon^{a_0m^2 + b_0m + c_0}} \widehat{P} && \textrm{by~(\ref{eq:TT_val_coeff})}\\
    & \quad\Longleftrightarrow\quad
    Z \textrm{ divides } \TT_{\varepsilon} \widehat{P} && \textrm{by~(\ref{eq:TT_val_val})}\\
    & \quad\Longleftrightarrow\quad
    \val(\widehat{P}(\varepsilon,Z))<\val(\widehat{P}(\varepsilon,0)).
\end{align*}
% The coefficients of $P_2$ are exactly the coefficients in this first
% $(a,b)$-block. If the constant coefficient in $Z$ occurs in that block,
% comparison of its valuation with the least valuation in the block is
% unchanged after adding the common offset $a_0m^2+b_0m$.
% {If it does not occur in that block, its valuation after substitution
% is either determined by a later block, and hence is strictly larger,
% or is infinite.}
% The same strict comparison holds for $P_2$, whose constant coefficient
% is then zero. These cases give the equivalence.
\end{proof}

\maintheoreduc*
\begin{proof}
We first describe the reduction and then prove its correctness and running-time
bounds.

Let $\bm f$ be an instance of $\APS_k$, where $k=\Q$ or
$k=\F_q$. By Proposition~\ref{prop-subsets-n-independent}, with arbitrarily
small constant error probability we may reduce to the case that
\[
    \bm f=(f_1,\ldots,f_{n+1})\in k[\bm X]^{n+1},
    \qquad \bm X=(X_1,\ldots,X_n),
\]
where $\bm f$ is minimally dependent. Over a finite field, the coefficients
of $\bm f$
may lie in a finite extension of $k$ of degree polynomial in the input
length. We enlarge this extension, if necessary, so that it also contains
the sampling set used below, and denote the resulting field by $\ell$.
For $k=\Q$ we take $\ell=\Q$.

Put
\[
    \Delta:=\max_{i\in[n+1]}\deg f_i,
    \qquad
    \Delta_{\mathrm{bad}}:=\Delta^{n+1}+n\Delta^n.
\]
Choose a finite set $S\subseteq\ell$ satisfying
\[
    |S|\geq 12(n+1)\Delta_{\mathrm{bad}}.
\]

We do not know which coordinate is the distinguished coordinate in
Proposition~\ref{prop-choice-ai-positive-image}. We therefore try all
$n+1$ possibilities. For each $j\in[n+1]$, we permute the coordinates so
that $f_j$ becomes the last coordinate, and write the resulting ordered
tuple again as
\((f_1,\ldots,f_n,f_{n+1})\).

Choose $\bm a=(a_1,\ldots,a_n)\in S^n$ uniformly at random and set
\[
    V:=\{\bm x\in\KK^n:
        f_i(\bm x)=a_i\varepsilon\text{ for all }i\in[n]\}.
\]

For this choice of $\bm a$, let $G_i$ be the homogenisation
of $f_i-a_i\varepsilon$ for $i\in[n]$, let $F$ be the homogenisation of
$f_{n+1}$, and form
\[
    \widetilde G_i:=G_i+TX_i^{d_i},
    \qquad
    L:=ZX_0^d+\sum_{i=1}^n W^iX_i^d-F
\]
and consider the resultant
\[
    R_0(T,W,\varepsilon,Z)
      :=\Res(\widetilde G_1,\ldots,\widetilde G_n,L)
\]
as above.
We note that the reduction does not compute $R_0$.

By multihomogeneity of the resultant, the degree of $R_0$ in each of
$T,W,\varepsilon$ is at most
\[
    D_0:=2(n+1)D^{n+1},
    \qquad
    D:=\max\{d_1,\ldots,d_n,d\}.
\]
Set
\(m:=D_0+1\).
The reduction supplies to the $\MRV_\ell$ oracle the
homogeneous forms obtained by making the substitutions 
\(T=\varepsilon^{m^2},
    W=\varepsilon^m
\) in
$\widetilde G_1,\ldots,\widetilde G_n,L$. By universality of the
resultant, the resultant of the obtained forms is
\[
    \widehat R_0(\varepsilon,Z)
       :=R_0(\varepsilon^{m^2},\varepsilon^m,\varepsilon,Z).
\]
We make the MRV query with the distinguished monomial $Z^0$. The reduction
accepts if and only if at least one of the $n+1$ oracle queries returns
yes.
This completes the description of the reduction.

We now argue the correctness of an individual query to the MRV oracle.
Suppose first that $V$ is finite and nonempty. By
Lemma~\ref{lem-resultant-images},
\[
    R_2(\varepsilon,Z)
      :=\TT_W\TT_T R_0(T,W,\varepsilon,Z)
      =c\prod_{\bm x\in V}
        \bigl(Z-f_{n+1}(\bm x)\bigr)^{m_{\bm x}},
\]
where $c\in\KK^\times$ and every $m_{\bm x}$ is positive. Hence, by
\eqref{eq-R-val},
\(\val(R_2(\varepsilon,Z))
       <\val(R_2(\varepsilon,0))\)
if and only if there exists $\bm x\in V$ such that
\(\val\bigl(f_{n+1}(\bm x)\bigr)>0\).

Lemma~\ref{lem-weighted-valuation} gives
\[
    \val(R_2(\varepsilon,Z))
       <\val(R_2(\varepsilon,0))
    \quad\Longleftrightarrow\quad
    \val(\widehat R_0(\varepsilon,Z))
       <\val(\widehat R_0(\varepsilon,0)).
\]
The condition on the right is that tested by
$\MRV_\ell$ when the distinguished monomial is $Z^0$. Hence,
whenever $V$ is finite and nonempty, the corresponding oracle query
returns yes if and only if
there exists $\bm x\in V$ such that $ \val\bigl(f_{n+1}(\bm x)\bigr)>0$.

Suppose that $\bm f$ is a positive APS instance. For the coordinate
ordering supplied by
Proposition~\ref{prop-choice-ai-positive-image}, there is a nonzero
polynomial $Q_{\mathrm{bad}}$ of degree at most
$\Delta_{\mathrm{bad}}-1$ such that, whenever
$Q_{\mathrm{bad}}(\bm a)\neq0$, the set $V$ is finite and nonempty and
contains a point $\bm x$ satisfying
\(\val\bigl(f_{n+1}(\bm x)\bigr)>0\).
For this coordinate choice, therefore, the corresponding $\MRV_\ell$
query returns yes. By Schwartz--Zippel,
\[
    \Pr[Q_{\mathrm{bad}}(\bm a)=0]
       \leq
       \frac{\Delta_{\mathrm{bad}}-1}{|S|}
       <\frac{1}{12(n+1)}.
\]
Thus, except with the stated small error probability, at least one of
the $n+1$ queries accepts.

Now suppose that $\bm f$ is a negative APS instance. We first observe that
the queries corresponding to all $n+1$ coordinate choices are well
behaved for generic $\bm a$.

Fix one such choice. Since $\bm f$ is minimally dependent, its
irreducible annihilator $H$ depends nontrivially on the coordinate left
unspecified. The construction in
Claims~\ref{claim:TWO}--\ref{claim:FOUR} yields a
nonzero polynomial $Q_j$ of degree at most
$\Delta_{\mathrm{bad}}-1$ such that
$V$ is finite and nonempty whenever 
\(Q_j(\bm a)\neq0\).

Assume that $Q_j(\bm a)\neq0$ for every coordinate choice $j$. If one of
the corresponding $\MRV_\ell$ queries returned yes, then
there would exist a point $\bm x\in\KK^n$ for
which the $n$ prescribed coordinates satisfy
\(f_i(\bm x)=a_i\varepsilon\)
and the remaining coordinate has positive valuation. This yields an
approximate solution to $\bm f$, contradicting the assumption that
$\bm f$ is a negative APS instance. Thus every query returns no under the assumption on $\bm a$.

A union bound and Schwartz--Zippel give
\[
    \Pr\bigl[
       Q_j(\bm a)=0
       \text{ for some }j\in[n+1]
    \bigr]
    \leq
    \frac{(n+1)(\Delta_{\mathrm{bad}}-1)}{|S|}
    <\frac{1}{12}.
\]

Combining the preceding bound with the error probability in
Proposition~\ref{prop-subsets-n-independent}, and, over finite fields,
the error incurred in the bounded-time construction and sampling of the
extension field, gives total error strictly less than $1/3$ after
choosing the corresponding constants appropriately.

It remains to check that the reduction runs in polynomial time. All
degree bounds and all exponents appearing in the substitutions
$T=\varepsilon^{m^2}$ and $W=\varepsilon^m$ have polynomially many bits.
The substituted homogeneous forms therefore have polynomial-size sparse
descriptions, even though their formal degrees may be exponential. Over
a finite field, the required extension $\ell/k$ has degree polynomial in
the input length and can be represented using polynomially many bits.
Finally, the reduction makes only $n+1$ oracle queries.

In summary, we obtain a randomised polynomial-time Turing reduction from
$\APS_{\Q}$ to $\mathrm{MRV}_{\Q}$, and, when $k$ is finite, from $\APS_k$ to
$\mathrm{MRV}_\ell$ for a finite extension $\ell/k$ of degree polynomial in the
input length.
\end{proof}

\section{MRV is in the Counting Hierarchy}
\label{sec-algorithm}
In this section we give the second substantial component of the proof of Theorem~\ref{theo-main}, namely an argument that $\mathrm{MRV}_k$
lies in $\CH$.
The main ingredient is~\cite[Theorem~6.7]{AGS26}, which 
places the bit language of the multivariate resultant in $\CH$ over the polynomial rings
$\Z[\bm Y]$ and $\F_{q}[\bm Y]$. We only need the following consequence of this result:

\begin{proposition}
\label{prop-coeff-resultant}
Let $k$ be $\Q$ or a finite field. Given homogeneous forms (in sparse representation)
\[
P_0,\ldots,P_n\in k[Y_1,\ldots,Y_v][X_0,\ldots,X_n]
\]
of specified positive degrees in $X_0,\ldots,X_n$, together with a binary-encoded exponent tuple $\bm\beta$, 
the problem of determining whether
\[
\coeff_{\bm Y^{\bm\beta}}
       {\Res_{X_0,\ldots,X_n}}(P_0,\ldots,P_n)=0
\]
is in $\CH$.
\end{proposition}
\begin{proof}
  Let $N$ be the input length.  We first consider the case that $k=\Q$.
  To set up the application of \cite[Theorem~6.7]{AGS26} we need to 
  clear the denominators in each input form.
  Notice that if $P_i$ is multiplied by a nonzero integer $d_i$,
  multihomogeneity of the resultant gives
\[
\Res(d_0P_0,\ldots,d_nP_n)
 =\left(\prod_{i=0}^n d_i^{D_i}\right)\Res(P_0,\ldots,P_n),
\]
where $D_i$ is the degree of the resultant in the coefficients of $P_i$.
Thus clearing denominators multiplies every coefficient of the resultant
by the same nonzero integer and does not change which
coefficients vanish. We can thus apply~\cite[Theorem~6.7]{AGS26}
over $\Z[\bm Y]$.

In the case $k=\F_{q}$, represented as in Section~\ref{sec-prelim},
we can directly apply \cite[Theorem~6.7]{AGS26}, which holds over
$\F_{q}[\bm Y]$.

It is important to note that the generic resultant circuit of~\cite{AGS26}, which underlies
the proof of~\cite[Theorem~6.7]{AGS26},
works with a dense coefficient representation of polynomials.
More precisely, by \cite[Lemma~5.12 and
Corollary~5.13]{AGS26}, an input gate may be indexed by
\((i,\alpha,\gamma)\), corresponding to the coefficient of
\(X^\alpha\bm Y^\gamma\) in \(P_i\). We need not explicitly build the
corresponding dense coefficient array: given such an index, its value
can be recovered in polynomial time by scanning the sparse description
of \(P_i\), returning zero if the monomial is absent. Thus the uniform
resultant circuits of \cite{AGS26} can be simulated from our sparse input
representation.

%However, as we now explain,
%this causes no difficulty for our sparse input representation. The
%circuits in~\cite[Lemma~5.12]{AGS26} compute the generic multivariate
%resultant: for forms of degrees \(d_0,\ldots,d_n\), their input gates
%are indexed by the formal coefficients \(u_{i,\alpha}\) of every
%possible monomial \(X^\alpha\) of degree
%\(d_i\). Moreover,~\cite[Corollary~5.13]{AGS26} treats the case in
%which these coefficients are themselves polynomials in auxiliary
%variables \(\bm Y\). In that case the input gates are indexed by
%triples \((i,\alpha,\gamma)\), representing the coefficient of
%\(X^\alpha\bm Y^\gamma\) in \(P_i\), and the circuit outputs the
%coefficients in \(\bm Y\) of the resultant.
%The input expected by the resultant circuit of~\cite{AGS26} can be
%simulated using indexed access to our sparse input. Specifically,
%given an input-gate index \((i,\alpha,\gamma)\), we scan the sparse
%representation of \(P_i\) and return the coefficient of
%\(X^\alpha\bm Y^\gamma\), returning \(0\) if the monomial is
%absent. This lookup takes time polynomial in the sparse input length
%\(N\), since both the exponent vectors and the coefficients have
%polynomial-length binary representations.

Standard multihomogeneous degree and height bounds imply that the
output coefficient representations of the resultant circuit
have at most $2^{\poly(N)}$ bits,
so their bit positions have polynomial-length addresses. A coefficient
is zero exactly when all of its indexed bits 
are zero. Universally quantifying over these polynomial-length addresses
preserves membership in $\CH$.
\end{proof}

\begin{corollary}
\label{cor-valuation-CH}
The problem $\mathrm{MRV}_k$ is in $\CH$ for $k=\Q$ or $k$ a finite field.
\end{corollary}
\begin{proof}
  We apply Proposition~\ref{prop-coeff-resultant} with parameter tuple $(\varepsilon,\bm Y)$. Write
\[
\mathcal R(\varepsilon,\bm Y)
:=\Res(P_0,\ldots,P_n)
=\sum_{\bm\beta,j}c_{\bm\beta,j}\varepsilon^j\bm Y^{\bm\beta}.
\]
A multihomogeneous resultant degree bound shows that every
relevant exponent is at most $2^{\poly(N)}$.
For the distinguished monomial $\bm Y^{\bm\beta_0}$, as observed in~\eqref{eq:cond},
the strict inequality in the definition of MRV is equivalent to
\[
\exists\bm\beta\;\exists j :
c_{\bm\beta,j}\neq0
\;\text{and}\;
\forall j'\in[0,j]\quad c_{\bm\beta_0,j'}=0.
\]
%Indeed, the first condition says that the minimum valuation among all
%$\bm Y$-coefficients is at most $j$, while the second says that the
%valuation of the distinguished coefficient is strictly greater than
%$j$. 
All quantified indices have polynomial-length binary encodings,
and each coefficient test is in $\CH$ by
Proposition~\ref{prop-coeff-resultant}. The fixed
existential--universal alternation therefore remains in $\CH$.
\end{proof}

\section{Relating APS to HN and ETR}
\label{sec-etr}
This final section gives a polynomial-time reduction of Hilbert's Nullstellensatz 
to Approximate Polynomial Satisfiability, over any characteristic.
We also show that $\APS_{\mathbb Q}$ reduces to the 
decision problem for the existential theory of real closed fields.
The proof uses the same valuation-amplification idea as in Proposition~\ref{prop:dense-representations}.

\begin{theorem}\label{thm:redu_HN_APS}
	For $k$ a finite field or the field of rational numbers,
    the Hilbert Nullstellensatz problem over $k$ is polynomial-time many-one
    reducible to approximate polynomial satisfiability:
	\begin{equation*}
	\HN_{k}\leq_m \APS_{k}.
	\end{equation*}
\end{theorem}

\begin{proof}
Let \(\bm p =(p_1(\bm x),\ldots,p_m(\bm x))\in k[x_1,\ldots,x_n]^m\) be an instance of $\HN_{k}$
and let $K$ denote the algebraic closure of $k$.
    Write \( \Delta := \max_i \deg p_i \) and set 
	\( B := \Delta^m \). Then $B$ is an effective Nullstellensatz degree bound (by, e.g.,~\cite[Theorem 1.1]{Jel05}), that is, if \( \bm 0 \not\in \bm p(K^n)\), then there exist polynomials \( q_1,\ldots,q_m \in K[x_1,\ldots,x_n]\) 
	such that	\(	1 = \sum_{i=1}^m q_i p_i \) and \(	\deg q_i\leq B\) for all \(i\in [m]\). 
	
	Define a polynomial map \( {\bm f} = (f_{i,j}) : K^n \to K^{m(n+1)}\) with the following components:
	\[
	f_{i,0}(\bm x)=p_i(\bm x),
	\qquad
	f_{i,j}(\bm x)=x_j^{B+1}p_i(\bm x) \qquad \textrm{with }1 
\le i\le m, 1\le j\le n.
	\]
	
	We claim that
	\begin{equation*}
	\bm 0 \in \bm p(K^n)
	\quad\iff\quad
	\bm 0\in\overline{{\bm f}(K^n)}.
	\end{equation*}
	
	Suppose first that there exists \( {\bm a} \in K^n \) such that $\bm p(\bm a)=\bm 0$.
	Then \(	f_{i,j}(a)=\bm 0 \) for every \(i\) and \(j\). Hence, \( \bm 0 \in {\bm f}(K^n) \subseteq \overline{{\bm f}(K^n)}\).
	
	For the converse direction, suppose that \( \bm 0 \in \overline{{\bm f}(K^n)}\). Using the valued-field formulation of \(\APS_k\) in terms of rational functions, there exists
	\[
	{\bm \beta}(\varepsilon)
	=
	(\beta_1(\varepsilon),\ldots,\beta_n(\varepsilon))
	\in \overline{K(\varepsilon)}^n
	\]
	such that the \(\varepsilon\)-valuation of all components of $\bm f(\bm \beta)$ is positive: \( \val\left(
	f_{i,j}({\bm \beta}(\varepsilon)) \right) >0 \) for all \(i\) and \(j\).
	
	Suppose for a contradiction that
	\[
	\bm 0 \not\in \bm p(K^n)\, .
	\]
	Then, by the effective Nullstellensatz, there are $q_i$ in \(K[\bm x]\) of degree at most \(B\) satisfying
	\begin{equation}\label{eq:eff_Nullstellensatz}
	\sum_{i=1}^m q_i p_i=1.
	\end{equation}
	
	Write \(s := \min_j \left( \val(\beta_j) \right)\) and let \(j_0\) be such that \(\val(\beta_{j_0}) = s\).
	For all \(i \le m\), we have
	\begin{equation}\label{eq:val_q}
		\val (q_i(\bm \beta)) \ge \min \left(0, B \cdot s \right).
	\end{equation}
	Moreover, we know that for all \(i\), 
    \begin{equation*} 
    \val \left( p_i({\bm \beta})\right) = \val \left( f_{i,0}({\bm \beta}) \right) >0 \quad \textrm{and}\quad
    \val \left( \beta_{j_0}^{B+1}p_i({\bm \beta})\right) = \val \left( f_{i,j_0}({\bm \beta}) \right) >0.
    \end{equation*}
    We deduce that
	\begin{equation}\label{eq:val_p}
		\val \left( p_i({\bm \beta}) \right) > \max \left( 0, -(B+1)s \right).
	\end{equation}
	From Equations~(\ref{eq:val_q}) and~(\ref{eq:val_p}), we get that 
	\begin{equation*}
		\val \left(p_i q_i ({\bm \beta}) \right) >0 \quad \textrm{for all }i.
	\end{equation*}
	However, the above bound implies that, after instantiating the identity~(\ref{eq:eff_Nullstellensatz}) at \(\bm \beta\), the left-hand side has positive valuation while the right-hand side has zero valuation. This is a contradiction.  We conclude that 
    $\bm 0 \in \bm p(K^n)$.  In summary, we have $\bm 0 \in \bm p(K^n) \Longleftrightarrow \bm 0 \in \overline{\bm f(K^n)}$.
	
	It remains to observe that the reduction can be carried out in polynomial time.  The critical fact is that 
    $\log (B+1)$ is polynomially bounded in the input size.  From this it is evident that representations of the $mn$ additional polynomials  $x_j^{B+1}p_i$ constructed in the reduction can be output in polynomial time.
	Consequently, \(\HN_{k} \leq_m \APS_{k}\).
\end{proof}

Recall that $\mathrm{ETR}$ is the set of true sentences of the form
\[
(\exists X_1,\ldots,X_n)\ \phi(X_1,\ldots,X_n),
\]
where $\phi$ is a quantifier-free formula of first-order logic
over the signature $(0,1,+,*,<)$, interpreted over a real closed field.

\begin{theorem}
$\APS_{\Q}$ is polynomial-time reducible to the decision problem for $\mathrm{ETR}$.
\end{theorem}
\begin{proof}
Let $\bm f= (f_1,\ldots,f_m) \in \mathbb Q[X_1,\ldots,X_n]^m$ 
be an instance of $\APS_{\mathbb Q}$.  Write $N$ for the description length.
Introduce vectors
$\UU=(U_1,\ldots,U_n)$ and $\VV=(V_1,\ldots,V_n)$ of real indeterminates, and 
let $g_r,h_r\in\Q[\UU,\VV]$ be such that
\[
{f_r(\UU+\mathrm i\VV)=g_r(\UU,\VV)+\mathrm i\,h_r(\UU,\VV)
\qquad(r\in[m]),}
\]
where $\mathrm i^2=-1$. Define
\[
S(\UU,\VV):={\sum_{r=1}^m
       \bigl(g_r(\UU,\VV)^2+h_r(\UU,\VV)^2\bigr)}.
\]
Then for all $\uu,\vv\in \mathbb R^n$, $S(\uu,\vv)$ equals
$\|\bm f(\uu+\mathrm i\vv)\|_2^2$. 

By Chevalley's theorem, 
$\bm f(\overline{\mathbb Q}^n)$ is a constructible subset of $\overline{\mathbb Q}^m$.  Since Euclidean and Zariski closures of a
constructible subset of $\overline{\mathbb Q}^m$ coincide, putting
\[
\alpha:=\inf_{\uu,\vv\in\RR^n}S(\uu,\vv),
\]
we have that $\bm f$ is a positive APS instance iff $\alpha=0$.

\begin{claim}
If $\alpha>0$, then $\alpha\geq2^{-2^{N^{O(1)}}}$.
\end{claim}
\begin{proof}
Consider the semialgebraic set
\[
\mathcal T:=\{t\in\RR:{(\exists\UU,\VV\in\RR^n)
                                      \ S(\UU,\VV)<t}\}.
\]
The predicate $S(\UU,\VV)<t$ can be represented by a formula of length
polynomial in the original input length $N$.
By Proposition~\ref{prop:dense-representations}, we may work with
polynomial-size arithmetic circuits for the input of the APS instance. Replacing
complex arithmetic in these circuits by arithmetic on real and imaginary
parts gives polynomial-size circuits for \(g_r\) and \(h_r\).
Introducing an existentially quantified variable for each internal gate
then yields a formula of polynomial length.

Since $\alpha$ is the distance between the set $\mathcal T$ and 0,
the claim follows from the quantitative separation bound~\cite[Corollary~3.4]{SchaeferStefankovic17}, which gives a lower bound of the form $2^{-2^{L+O(1)}}$ for two semialgebraic sets of positive distance that have description size at most $L$.
\end{proof}

By the claim there exists an integer constant $c>0$ such that   
\[
\alpha=0\quad\text{iff}\quad
\alpha<2^{-2^{N^c}}=:\tau.
\]
The value $\tau$ has a short description using iterated squaring. We can then write 
an existential formula for the predicate $\alpha=0$ as follows: 
%indeterminates $T_0,\ldots,T_{\rho(N)}$, our formula is
\[
\exists\UU,\VV,T_0,\ldots,T_{N^c}\quad
\left(
2T_0=1\ \land\
\bigwedge_{j=0}^{N^c-1}T_{j+1}=T_j^2\ \land\
S(\UU,\VV)<T_{N^c}
\right)\, .
\]
In summary, $\bm f$ is a positive instance of $\APS_\mathbb Q$ if and only if
$\alpha=0$, and $\alpha=0$ if and only if the above is a true formula of $\mathrm{ETR}$.
\end{proof}

\bibliographystyle{alpha}
\bibliography{ref1}

\end{document}